\documentclass{amsart}
\usepackage{graphicx,enumerate,framed,amsmath}

\usepackage{color,hyperref}
\definecolor{dark-blue}{rgb}{0.15,0.15,0.4}
\definecolor{dark-red}{rgb}{0.4,0.15,0.15}
\definecolor{medium-red}{rgb}{0.6,0,0}
\definecolor{medium-blue}{rgb}{0,0,0.6}
\hypersetup{
    colorlinks=true,
    linktoc=all,
    citecolor=medium-red,
    filecolor=medium-blue,
    linkcolor=medium-blue,
    urlcolor=medium-blue,
}
        \theoremstyle{plain} %--default
        \newtheorem{theorem}             {Theorem}  [section]
        
        \newtheorem{corollary}  [theorem]{Corollary}
        \newtheorem{proposition}[theorem]{Proposition}

        \theoremstyle{definition}
        \newtheorem{definition} [theorem]{Definition}
        
        \newtheorem{example}    [theorem]{Example}

        \theoremstyle{remark}
        \newtheorem{remark}              {Remark}
\usepackage{color,hyperref}
\definecolor{dark-blue}{rgb}{0.15,0.15,0.4}
\definecolor{dark-red}{rgb}{0.4,0.15,0.15}
\definecolor{medium-red}{rgb}{0.6,0,0}
\definecolor{medium-blue}{rgb}{0,0,0.6}
\hypersetup{
    colorlinks=true,
    linktoc=all,
    citecolor=medium-red,
    filecolor=medium-blue,
    linkcolor=medium-blue,
    urlcolor=medium-blue,
}
                \numberwithin{equation}{section}
\makeatletter
\renewcommand*\env@matrix[1][\arraystretch]{%
  \edef\arraystretch{#1}%
  \hskip -\arraycolsep
  \let\@ifnextchar\new@ifnextchar
  \array{*\c@MaxMatrixCols c}}
\makeatother

\begin{document}
% Fakesection
%\listofchanges%[style=summary]%[style=list]

%\title[Canonoid Transformations and BiHamiltonian Structures]{Linear Transformations, Canonoid Transformations and BiHamiltonian Structures}
\title[Canonoid and Poissonoid transformations]{Canonoid and Poissonoid transformations, symmetries and biHamiltonian structures}
\author[]{%
	Giovanni Rastelli\\
	Manuele Santoprete\\
}

\address{Dipartimento di Matematica\\
           Universit\'a di Torino\\
           Torino, via Carlo Alberto 10, Italia\\
          }
\email{ giovanni.rastelli@unito.it}
\address{Department of Mathematics\\
           Wilfrid Laurier University\\
           75 University Avenue West\\
           Waterloo, ON,Canada\\
           }%
\email{msantoprete@wlu.ca}   


%\date     {July 2, 1991}

%         \thanks will become a 1st page footnote.
%\thanks{The authors wish to thank Luca Degiovanni for comments and suggestions. The second author was supported  by an NSERC Discovery Grant.}
%\dedicatory{}
 
\begin{abstract}
    We give a characterization of linear canonoid transformations on symplectic manifolds and we use it to generate biHamiltonian structures for some mechanical systems. Utilizing this characterization we also study the behavior of the harmonic oscillator under canonoid transformations. We present a description of canonoid transformations due to E.T. Whittaker, and we show that it leads, in a natural way, to  the  modern, coordinate-independent definition of canonoid transformations. We also generalize canonoid transformations to Poisson manifolds by introducing Poissonoid transformations. We give examples of such transformations for   Euler's equations  of the rigid body (on $ \mathfrak{ so}^\ast (3) $ and $ \mathfrak{ so}^\ast (4)$)  and for an integrable case of  Kirchhoff's equations for the motion of a rigid body immersed in an ideal fluid. We study the relationship between biHamiltonian structures and Poissonoid transformations  for these examples. We analyze the link between Poissonoid transformations, constants of motion, and symmetries.

%Linear canonoid transformations on symplectic manifolds are employed to generate biHamiltonian structures for some mechanical systems in dimension $2\times 2$. The behavior of  the quadratic superintegrable structure of the harmonic oscillator and other systems under canonoid transformations is investigated. An invariant characterization of canonoid transformations, foreseen by E. T. Whittaker, is developed.  Canonoid transformations are generalized to Poissonoid transformations on Poisson manifolds  and   examples of their applications are given in rigid body dynamics for Euler,  Kirchhoff and Clebsh systems. 
\end{abstract}
\maketitle
\setcounter{tocdepth}{1}
\tableofcontents
\section{Introduction }

%In this paper we study several important examples of Hamiltonian systems on a finite-dimensional symplectic or Poisson  manifold $M$, and we show that, for each example, another Hamiltonian structure can be  obtained by using certain special linear transformations on $M$. In other words, we essentially show that such examples posses a biHamiltonian structure.

 BiHamiltonian systems are, in a nutshell, dynamical systems described by a vector field that is  Hamiltonian  with respect to two  distinct Poisson (or symplectic) structures and  two  associated  (possibly distinct) Hamiltonian functions.   Under certain additional hypothesis, possessing a biHamiltonian structure is enough to guarantee the integrability of the system (see for example \cite{magri_eight_2004}). During the last few decades it has been shown that many integrable systems are in fact biHamiltonian, consequently,  biHamiltonian structures are now an important paradigm for understanding integrability.
In some cases, a new Hamiltonian structure can be obtained with a transformation of coordinates. This may be possible when, on a symplectic manifold, the transformation changes the Hamiltonian characterization of a Hamiltonian vector field. Such transformations (in the case the symplectic manifold is $ \mathbb{R}^{2n}  $ and the symplectic form the standard one) were dubbed ``canonoid"  and popularized by Saletan and Cromer \cite{saletan_theoretical_1971}, and by Currie and Saletan \cite{currie_canonical_1972}, but they were know well before the 1970s, in fact, they were already present in  the 1904  edition of the classical book of Whittaker \cite{whittaker_treatise_1988}. This type of transformations include the well known canonical ones.
The main difference between these   transformations  is that, while the canonoid ones are specific to the problem considered, the canonical ones preserve the Hamiltonian form of every Hamiltonian system on the manifold, and  leave invariant the symplectic structure. Therefore, canonical transformations  cannot be used to  generate different symplectic structures. Strictly canonoid transformations (i.e., those canonoid transformations that are not canonical), in contrast, change the symplectic structure and only preserve the Hamiltonian form of some chosen Hamiltonian systems, and thus can be used to  generate different symplectic structures. Canonoid transformations are the argument of about 20 papers, among them, we cite here the ones more closely related to the content of our article. A first set of papers concerns primarily the characterization of canonoid transformations and their relations with canonical transformations: \cite{currie_canonical_1972,negri_canonoid_1987,carinena_canonoid_1988,vegas_canonical_1989,carinena_generating_1990,carinena_canonoid_2013}. A second set deals primarily with applications of canonoid transformations to the analysis of Hamiltonian systems: \cite{tempesta_quantum_2002,landolfi_certain_2007}.

In this paper we use a modern geometrical definition of canonoid transformation based on locally Hamiltonian vector fields. This definition coincides to the so called quasi-canonical transformations of Marmo \cite{marmo_equivalent_1976} and  reduces to  the  definition of Saletan and Cromer  \cite{saletan_theoretical_1971}  in the simplest case of a topologically trivial system, or at least when considering only local expressions for the system.  By generalizing the approach of  \cite{fasano_analytical_2006}, we obtain simple explicit conditions for linear canonoid transformations on $ \mathbb{R}^n  $. We use this method to analyze some  examples, including the harmonic oscillator in $ \mathbb{R}  ^{ 4 } $. We also recall the approach of Whittaker \cite{whittaker_treatise_1988} and show that the modern definition of canonoid  transformation we employ follows naturally from such approach.  
Moreover, we  extend this type of transformations to the case of Poisson manifolds, by introducing a generalization of the canonoid transformations that we dub Poissonoid transformations. This type of transformations, as far as we know, have not been studied before, and they  allow us to find  biHamiltonian structures in the case of Poisson manifolds. The Casimirs of the new Poisson structures found this way provide first integrals of the Hamiltonian system. Furthermore, if the Poisson structures are compatible, the integrability of the systems follows from the theory of biHamiltonian systems. 
We also study the relationship between  Linear Poissonoid transformations and biHamiltonian structures  in some examples, namely Euler's equations  for the rigid body (on $ \mathfrak{ so}^\ast (3) $ and $ \mathfrak{ so}^\ast (4)$)  and an integrable case of Kirchhoff's equations for the motion of a rigid body immersed in an ideal fluid.  
We conclude with a study of the relations among infinitesimal Poissonoid transformations, Noether theorem and master symmetries, generalizing to the Poisson case some results obtained in \cite{carinena_canonoid_2013} for canonoid transformations.
%We also give several examples of  Linear Poissonoid transformations and we use them to find biHamiltonian structures for  Euler's equations  for the rigid body (on $ \mathfrak{ so}^\ast (3) $ and $ \mathfrak{ so}^\ast (4)$)  and for Kirchhoff's equations for the motion of a rigid body immersed in an ideal fluid.  

Our aim is to provide to the non-specialists  an  introduction to canonoid transformation and  biHamiltonian systems through the analysis of several examples. For the specialists, we highlight the new definition of Poissonoid transformations, the role played by simple linear canonoid (and Poissonoid) transformations in the determination of several biHamiltonian structures, and the relationship between Poissonoid transformations, integrals of motion, and symmetries.

 The paper is organized as follows. In section 2 we recall some essential facts concerning Poisson geometry, symplectic geometry and biHamiltonian structures, and we set the notations   employed in the rest of the article. Section 2 can be skipped by readers already familiar with these topics. In section 3  we introduce canonoid transformations on symplectic manifolds and we study examples of linear canonoid transformations.  In section 4  we analyze how the superintegrable structure of some simple systems behaves under linear canonoid transformations.  In  section 5 we translate into more modern language  the characterization  of canonoid transformations given in Whittaker \cite{whittaker_treatise_1988}. In section 6, we extend  the idea of canonoid transformations to Poisson manifolds by introducing Poissonoid transformations, and we give several examples of such transformations. In the last section we analyze the link between infinitesimal Poissonoid transformations, master symmetries and constants of motion.  % Moreover, we give several examples of  Linear Poissonoid transformations, in particular we use them to study  Euler's equations  for the rigid body (on $ \mathfrak{ so}^\ast (3) $ and $ \mathfrak{ so}^\ast (4)$)  and Kirchhoff equations for the motion of a rigid body immersed in an ideal fluid.   

 \section{Poisson, symplectic and biHamiltonian structures}
 \subsection{Poisson structures}
We now recall the fundamental  definitions  and some of  the main results concerning Poisson structures, for a more detailed account we refer the reader to the following references: \cite{marsden_introduction_1999,laurent-gengoux_poisson_2012,choquet-bruhat_analysis_2000,libermann_symplectic_1987,rudolph_differential_2012}.

\begin{definition} Let $M$ be a smooth manifold, and let $ C ^{ \infty }  (M) $ be the set of smooth functions on $M$. A {\bf Poisson bracket}  or {\bf Poisson structure}  is a skew-symmetric bilinear operation $ \{ \cdot , \cdot \} :C ^{ \infty }  (M) \times C ^{ \infty }  (M) \to C ^{ \infty }  (M)$ which satisfies the Jacobi identity 
    \[
        \{ \{ F, G \} , H \} + \{ \{ G,H \} , F \} + \{ \{ H,F \} , G \} = 0 
    \]
    and  the Leibnitz identity
    \[
        \{ f, gh \} = \{ f , g \} h + g \{ f , h \}.
    \]
    Associated with the bracket there is a bivector field defined by
    \[ \{ F, G \} (x) = \pi (x) (\mathbf{d} F (x)  , \mathbf{d} G (x)),\]
    called {\bf Poisson bivector } or {\bf Poisson tensor} field.   The pair $ (M, \pi ) $  is called {\bf Poisson manifold}.
 
\end{definition} 
 
%\begin{definition}
%	A bivector field $ \pi $ on a smooth manifold $M$ satisfying $ [\pi , \pi ]=0 $ (where $ [,] $ is the
%	Schouten bracket )  is called a {\bf Poisson bivector } field, or {\bf Poisson structure}. The pair $ (M, \pi ) $  is called {\bf Poisson manifold}.
%\end{definition}
%The bracket defined by $ \{ F, G \} = \pi ( \mathbf{d} F , \mathbf{d} G ) $ is called Poisson bracket.
Let $\alpha$ and $\beta$ be differential forms, than we define the map
$ \pi ^\sharp : T ^\ast M \to T M$  as $ \pi ( \alpha , \beta ) = \left\langle \alpha , \pi ^\sharp \cdot \beta \right\rangle $.
The  {\bf rank } of $\pi$ at $ x \in M $ is the rank of the linear map $ \pi _x ^\sharp : T ^\ast _x M \to T  _x M $. In general, the rank will vary from point to point.

\begin{definition}
	A {\bf regular point}  of a Poisson manifold is a point where the rank of the Poisson bivector is locally constant, the remaining points are called  {\bf singular points}.
	A {\bf regular Poisson bivector} $ \pi $ is a Poisson bivector whose rank is constant. Similar definitions apply to Poisson structures. A {\bf regular Poisson manifold } is a Poisson manifold endowed with a regular Poisson structure.
\end{definition}

\begin{definition} Let $ (M , \pi) $ be a Poisson manifold. A smooth, real valued function $ C: M \to \mathbb{R}  $ is called a {\bf Casimir function } if the Poisson bracket of $C$ with any other real-valued function vanishes identically, i.e. $ \{ C, H \} = 0 $ for all $ H: M \to \mathbb{R}  $.  
\end{definition}

An alternative way of introducing the Poisson bivector uses the so called Schouten-Nijenhuis bracket, namely   an extension of the Lie bracket of vector fields to skew-symmetric multivector fields, see \cite{marsden_introduction_1999,vaisman_lectures_1994}).
\begin{proposition}  
    A bivector field $\pi$   on $M$ is the Poisson bivector of a Poisson structure on $M$ if and only if
 $ [\pi , \pi ] = 0 $, where $ [\cdot , \cdot ] $ is the Schouten-Nijenhuis  bracket.
\end{proposition}

\begin{definition}
	Let $ ( M , \pi ) $ be a Poisson manifold. The  {\bf Hamiltonian vector field } of  a smooth function $H: M \to M$  is the vector field such that
	\[\mathcal{X} _H [F] = \{ F , H \}  = \left\langle \mathbf{d} F, \pi ^\sharp \cdot  \mathbf{d} H \right\rangle \]
	for every smooth function $F$ on $M$. The function $H$ is called  {\bf Hamiltonian function}. The triplet $ (M , \pi , H) $ is called {\bf Hamiltonian system}. 
\end{definition}

The definition of Hamiltonian system can be generalized to systems that are Hamiltonian in a neighborhood of each point of the manifold $ (M , \pi) $.  
Here we use a definition of locally Hamiltonian given in \cite{laurent-gengoux_poisson_2012}, note that this definition differs from the one used in \cite{libermann_symplectic_1987} and \cite{choquet-bruhat_analysis_2000}.

\begin{definition}
	A vector field $\mathcal{X}$ on a Poisson  manifold $ ( M , \pi  ) $ is  called {\bf locally Hamiltonian } if for every $ x \in M $ there is a neighborhood $U$ of $x$ and a smooth function $ H _U $ defined on this neighborhood such that $ \mathcal{X} = \pi ^\sharp \cdot  \mathbf{d} H _U $, that is $ \mathcal{X} $ is Hamiltonian in $U$ with the locally defined Hamiltonian $ H _U $. A triplet $ (M , \pi , \mathcal{X}) $ as above is called a {\bf locally Hamiltonian system}. 
\end{definition}

Let $ ( M , \pi ) $ be a  Poisson manifold of dimension $d$. In  a  neighborhood $ U $ of a point $ p \in M $, with local coordinates $  \mathbf{x} = ( x ^1  , \ldots , x ^d  ) $, the bivector field $ \pi $ can be written as
\[
	\pi = \sum_{ 1 \leq i < j \leq d} \{ x ^i  , x ^j  \} \frac{ \partial } { \partial x ^i  } \wedge \frac{ \partial } { \partial x ^j  }
\]
and the vector field
\[
	\mathcal{X} _H  = \sum_{ i, j = 1 } ^d  \{ x ^i  , x ^j  \} \frac{ \partial H } { \partial x ^j  } \frac{ \partial } { \partial x ^i  }.
\]

We recall that a function $ F \in C ^{ \infty } (M) $ such that $ \{ F, H \}  = 0 $ is called a {\bf constant of motion} or first integral of the Hamiltonian system. 
\begin{definition} 
    Let $ (M , \pi) $ be a Poisson manifold, and let $ (M , \pi , H) $ be a Hamiltonian system on $M$ with Hamiltonian vector field $ \mathcal{X} _{ H } $.   
    A vector field $\xi  $ on $M$ is called an  {\bf infinitesimal symmetry  of the vector field} $ \mathcal{X} _H $ if $ \mathcal{L} _{ \xi } \mathcal{X} _{ H } = 0$. Moreover,  $ \xi $  is called a  {\bf  Poisson infinitesimal symmetry } of $ (M , \pi , H) $ if it is an infinitesimal symmetry of both $ \pi $ and $ H $, that is if 
    \[ \mathcal{L} _{ \xi  } \pi = 0 \quad \mbox{ and } \quad  \mathcal{L} _{ \xi  } H = \xi [H] = 0.\]
     A Poisson infinitesimal symmetry  $ \xi   $  is called  a {\bf (locally) Hamiltonian infinitesimal symmetry  } of $ (M , \pi , H) $ if, in addition, $ \xi  $ is (locally) Hamiltonian. That is there is a (locally defined) function $K$ such that $ \xi  = \pi ^\sharp \cdot \mathbf{d} K $. 
\end{definition}	

Note that, in the non-degenerate (symplectic) case, Poisson infinitesimal symmetries coincide with locally Hamiltonian infinitesimal symmetries.

It is important to keep in mind there is a distinction  between Poisson infinitesimal symmetries  of   $ (M , \pi , H) $ and infinitesimal symmetries of the Hamiltonian vector field $ \mathcal{X} _H = \pi ^\sharp \cdot \mathbf{d} H $. The following proposition clarifies the relationship  between the two types of symmetries

\begin{proposition}  Let $ (M , \pi) $ be a Poisson manifold, and let $ (M , \pi , H) $ be a Hamiltonian system on $M$ with Hamiltonian vector field $ \mathcal{X} _{ H } $. If  $ \xi $  is a    Poisson infinitesimal symmetry  of $ (M , \pi , H) $ then it is also an  infinitesimal symmetry  of the vector field $ \mathcal{X} _H $.% If  $ \pi $ is regular and  $ \xi $ is such that $ \mathcal{L} _{ \xi } \pi= 0 $ and $ [\xi , \mathcal{X} _H ] = 0 $ then $ \xi $ is a Hamiltonian infinitesimal symmetry. 
\end{proposition}  
\begin{proof} 
Let $ \alpha $ be an arbitrary $1$-form. Since $ \mathcal{X} _H =  \pi ^\sharp \cdot \mathbf{d} H $ we have 
\begin{equation} \begin{aligned}\label{eqn:symmetry} 
     0& =   \mathcal{L} _{ \xi } \left\langle \alpha , \mathcal{X} _H - \pi  ^\sharp \cdot \mathbf{d} H \right\rangle  = \mathcal{L} _{ \xi } \left\langle \alpha , \mathcal{X} _H \right\rangle - \mathcal{L} _{ \xi } \pi (\alpha , \mathbf{d} H)  \\
     & = \left\langle \mathcal{L} _{ \xi } \alpha , \mathcal{X} _H \right\rangle + \left\langle \alpha , \mathcal{L} _{ \xi } \mathcal{X} _H \right\rangle - (\mathcal{L} _{ \xi } \pi) (\alpha , \mathbf{d} H) - \pi (\mathcal{L} _{ \xi } \alpha , \mathbf{d} H) - \pi (\alpha, \mathcal{L} _{ \xi } \mathbf{d} H)\\
     & = \left\langle \mathcal{L} _{ \xi } \alpha , \mathcal{X} _H \right\rangle + \left\langle \alpha , \mathcal{L} _{ \xi } \mathcal{X} _H \right\rangle - (\mathcal{L} _{ \xi } \pi) (\alpha , \mathbf{d} H) - \left\langle \mathcal{L} _{ \xi } \alpha , \pi ^\sharp \cdot (\mathbf{d} H)\right\rangle  -  \left\langle \alpha, \pi ^\sharp \cdot (\mathcal{L} _{ \xi } \mathbf{d}  H)\right\rangle \\
    &=   \left\langle \alpha , \mathcal{L} _{ \xi } \mathcal{X} _H \right\rangle - (\mathcal{L} _{ \xi } \pi) (\alpha , \mathbf{d} H)  -  \left\langle \alpha, \pi ^\sharp \cdot ( \mathbf{d}( \mathcal{L} _{ \xi } H))\right\rangle.
\end{aligned} \end{equation}  
    If $ \xi $ is a Poisson infinitesimal symmetry of $ (M, \pi , H) $, then $\mathcal{L} _{ \xi  } \pi = 0$  and  $ \mathcal{L} _{ \xi  } H =  0$. Since $ \alpha $  is arbitrary, by the equation above we have $  \mathcal{L} _{ \xi } \mathcal{X} _H = 0 $.
\end{proof}

The converse of the proposition above is  clearly not true in general, even when $ \pi $ is non-degenerate \cite{carinena_canonoid_2013}.  

Note that Hamiltonian infinitesimal symmetries are very important because, through Noether's theorem (see below), they give rise to constants motion, which are very useful in the process of reduction the Hamiltonian system.   
\begin{theorem}[Noether's theorem] Let $ (M , \pi , H) $ an Hamiltonian system. If $F$ is a constant of motion then its vector field is a Hamiltonian infinitesimal symmetry. Conversely, each Hamiltonian infinitesimal symmetry is the Hamiltonian vector field of a constant of motion, which is unique up to a (time dependent) Casimir function.  
\end{theorem} 
\begin{proof}
   Suppose $F$ is a constant of motion. Then  $ \mathcal{X} _F = \pi ^\sharp \cdot \mathbf{d} F $ is an Hamiltonian vector field, so that $ \mathcal{L} _{ \mathcal{X} _F } \pi = 0 $. Moreover, since $ \mathcal{X}_H [F]= 0 $ , we have that
   \[0=\mathcal{X}_H [F] = \{ F, H \} = - \{ H, F \} = \mathcal{X} _F [H]=0.  \]
  Thus $ \mathcal{X} _F $ is a Hamiltonian infinitesimal symmetry.
  Now suppose that $V $ is a Hamiltonian infinitesimal symmetry of $ (M , \pi , H) $. Since $ V $ is Hamiltonian there is a function $F$ such that $ V = \mathcal{X} _F = \pi ^\sharp \cdot \mathbf{d} F $. Since it is an infinitesimal  symmetry of $ (M , \pi , H) $  we have that 
  \[
      0=\mathcal{X} _F [H]= \{ H , F \} = - \{ F , H \} = \mathcal{X} _H [F]
  \]
  so $F$ is a constant of motion. 
  If $ \tilde F $ is another function that satisfies $\mathcal{X} _{ \tilde F } = V = \mathcal{X} _{ F } $, then 
      \[
         0 =  \mathcal{X} _{ \tilde F } - \mathcal{X} _F = \pi ^\sharp \cdot  (\mathbf{d} (\tilde F - F) )   
      \]
      so $ \tilde F - F $ must be a Casimir of $ \pi $, since if we apply the above to an arbitrary function $G$ we have that 
      \[0 = \left\langle\mathbf{d} G,\pi ^\sharp \cdot  \mathbf{d} (\tilde F - F ) \right \rangle = \pi (\mathbf{d} G, \mathbf{d} (\tilde F - F)) = \{G, \tilde F - F  \} .  
      \]
  \end{proof}

\begin{definition}
	A vector field  $ \mathcal{X} $ on a Poisson  manifold $ ( M , \pi  ) $ is called a {\bf Poisson vector field }  iff $ \mathcal{L} _\mathcal{X} \pi = 0 $.
\end{definition}

In particular it follows that any locally Hamiltonian vector field is Poisson:
\begin{proposition}
	If $ \mathcal{X} $ is a locally Hamiltonian vector field on a Poisson manifold $ (M , \pi) $, then   it is a Poisson vector field.
\end{proposition}
\begin{proof}
	See \cite{laurent-gengoux_poisson_2012}.
\end{proof}
The converse is not true in general. For example, if the Poisson structure is trivial, then any vector field is Poisson, while the only Hamiltonian vector field is the trivial one. In the special case of a symplectic manifold, a vector field is Poisson if and only if it is locally  Hamiltonian (see Proposition \ref{prop:locally_Hamiltonian}).

\begin{theorem}[Weinstein's splitting theorem ]
	Let $ (M , \pi) $ be a Poisson manifold, let $ x \in M $ be an arbitrary point and denote the rank of $ \pi $  at $x$ by $ 2r $. There exists a coordinate neighborhood $U$ of $x$ with coordinates  $ (q ^1  , \ldots q ^r  , p _1 , \ldots, p _r , z ^1  , \ldots z ^s ) $ centered at $x$, such that, on $U$,
	\begin{equation}\label{eqn:splitting_theorem}
		\pi = \sum _{ i = 1 } ^r \frac{ \partial } { \partial q ^i  } \wedge \frac{ \partial } { \partial p _i } + \sum _{ 1 \leq k <l \leq s } \phi ^{ kl } (z) \frac{ \partial } { \partial z ^k  } \wedge \frac{ \partial } { \partial z ^l  },
	\end{equation}
	where the functions $ \phi ^{ k l } $ are smooth functions which depend on $ z = (z ^1  , \ldots , z ^s) $ only, and which vanish when $ z = 0 $. Such local coordinates are called {\bf splitting coordinates}, centered at $x$.
   
    In particular, if there is a neighborhood $V$ of $x$ such that the rank is constant and equal to $ 2r $, then there exists a coordinate neighborhood $U$ of $x$ with coordinates  $ (q ^1  , \ldots q ^r  , p _1 , \ldots, p _r , z ^1  , \ldots z ^s ) $ centered at $x$, such that, on $U$,
    	\begin{equation}%\label{eqn:splitting_theorem}
		\pi = \sum _{ i = 1 } ^r \frac{ \partial } { \partial q ^i  } \wedge \frac{ \partial } { \partial p _i } 
	\end{equation}
    	Moreover, $\pi$ is locally of the form above, in terms of arbitrary splitting coordinates on $M$. Such coordinates are called {\bf Darboux coordinates}.

\end{theorem}
\begin{proof}
	See \cite{laurent-gengoux_poisson_2012}.
\end{proof}
\begin{remark}
	For a given point $x$ on a Poisson manifold $M$, splitting coordinates are not unique. The Poisson structure, which is defined in a neighborhood of $ z = 0 $ by the second term of \eqref{eqn:splitting_theorem}, however, is unique up to a Poisson diffeomorphism (see for example \cite{laurent-gengoux_poisson_2012}).
\end{remark}
%\begin{theorem}[Darboux' theorem]
%	Let $ (M ,\pi) $ be a Poisson manifold of dimension $d$, and suppose $x$ is a point where the rank is locally constant (that is there is a neighborhood $V$ of $x$ such that the rank is constant) and   equal to $ 2r $. There exists a coordinate neighborhood $U$ (with $ U \subset V $) of $x$ with coordinates $ (q ^1  , \ldots q ^r  , p _1 , \ldots p _r , z ^1  , \ldots z ^s ) $ such that, on $U$,
%	\[
%		\pi = \sum _{ i = 1 } ^r \frac{ \partial } { \partial q ^i  } \wedge \frac{ \partial } { \partial p _i }.
%	\]
%	Moreover, $\pi$ is locally of the form above, in terms of arbitrary splitting coordinates on $M$. Such coordinates are called Darboux coordinates.
%\end{theorem}

\begin{proposition} \label{prop:Hamilton's equations I}   Let $ (M , \pi ) $ be a Poisson manifold and let $ x \in M $. Suppose  $\mathcal{X} $ is a locally Hamiltonian vector field. Let $U$ be a neighborhood of $x$ with splitting coordinates  $(\mathbf{q} , \mathbf{p}, \mathbf{z} ) =   (q ^1  , \ldots , q ^r  , p _1 , \ldots , p _r, z ^1  , \ldots , z ^s  ) $, and $ H _U $ such that $ \mathcal{X} = \pi ^\sharp \cdot  \mathbf{d} H _U  = \pi ( \cdot, \mathbf{d} H _U ) $. In these coordinates
	\[
		\mathcal{X} = \begin{bmatrix}[1.3]
		\frac{ \partial H _U  }{ \partial \mathbf{p}  }  \\
		- \frac{ \partial H _U  }{ \partial \mathbf{q}  }\\
		\Phi ^T \frac{ \partial H _U  }{ \partial \mathbf{z}  }
		\end{bmatrix}
	\]
	where $ \Phi $ is the matrix of entries $ [\Phi ]_{ k l } = \tilde \phi ^{ k l }(z) $ (where  $ \tilde \phi ^{ k l }(z)=\phi ^{ k l }(z)$ for $ k<l $, $  \tilde \phi ^{ k l }(z)= -   \phi ^{ lk }(z)$ for $ k> l $ and $  \tilde \phi ^{ k l }(z)=0 $ for $ k = l $).
	% \[
	%     \mathcal{X} = \left(\frac{ \partial H _U  } { \partial \mathbf{p}  } ,- \frac{ \partial H _U  } { \partial \mathbf{q}  },0 \right) .
	% \]
	
	Thus $ (\mathbf{q}  (t) , \mathbf{p}  (t), z ( t ) ) $ is an integral curve of $ \mathcal{X} $ if an only if {\bf Hamilton's equations} hold:
	\[
		\dot q ^i  = \frac{ \partial H _U  } { \partial p _i } , \quad \dot p _i = -\frac{ \partial H _U  } { \partial q ^i  } , \quad \dot z ^k  = \sum _l \tilde \phi ^{ k l }(z) \frac{ \partial H _U  }{ \partial z ^l  }
	\]
	for $ i = 1 , \ldots , r $, and $ k = 1, \ldots , s $, where the functions $ \tilde \phi ^{ i j } (z) $ depend on the choice of splitting coordinates.
	Moreover, if the rank is locally constant at the point $x$, then the vector field, written  in Darboux coordinates, is
	\[
		\mathcal{X} = \begin{bmatrix}[1.3]
		\frac{ \partial H _U  }{ \partial \mathbf{p}  }  \\
		- \frac{ \partial H _U  }{ \partial \mathbf{q}  }\\
		0
		\end{bmatrix}.
	\]
	and Hamilton's equations take the simpler form
	\[
		\dot q ^i  = \frac{ \partial H _U  } { \partial p _i } , \quad \dot p _i = -\frac{ \partial H _U  } { \partial q ^i  } , \quad \dot z ^k  =0
	\]
	for $ i = 1 , \ldots , r $, and $ k = 1, \ldots , s $.
\end{proposition}

\begin{proof}
    See \cite{laurent-gengoux_poisson_2012}.
\end{proof}

\subsection{Symplectic structures}
%When $dim M=2r$ and the Poisson structure is everywhere regular and of maximal rank $=2r$, the Poisson structure is called {\bf symplectic}.

%--RELAZIONI TRA $\pi$ ED $\omega$ ----

We now give a brief account of symplectic structures, for more details see \cite{abraham_foundations_1978,marsden_introduction_1999}.
\begin{definition}
   A {\bf symplectic form } (or {\bf symplectic structure})  on a manifold $M$ is a nondegenerate, closed two-form $ \omega $ on $M$. A  {\bf symplectic manifold } $ (M , \omega) $ is a manifold $M$ together with a symplectic form $ \omega $ on $M$.    
\end{definition} 
\begin{definition} Let $ ( M , \omega ) $ be a symplectic manifold, and $\mathcal{X}$ a vector field on $M$. If there is a smooth function $ H : M \to \mathbb{R}  $ such that
	\[ \mathbf{i} _\mathcal{X} \omega = \mathbf{d} H\]
	we say that $\mathcal{X}$ is a  {\bf Hamiltonian vector field }.
\end{definition}

\begin{definition} Let $ (M , \omega) $ be a symplectic manifold. The {\bf Poisson bracket associated with } $ \omega $ is defined by
    $  \{ F , G \}  = \omega( \mathcal{X} _F  , \mathcal{X}  _G ).$ \end{definition} 
From the above definition it follows that, associated to $ \omega $ there is a Poisson bivector $ \pi $, so that a symplectic structure is a regular Poisson structure of maximal rank.  The basic link between the Poisson bivector $ \pi $ and the symplectic form $ \omega $ is that they are associated to the same Poisson bracket
\[
    \{ F, H \} = \pi (\mathbf{d} F , \mathbf{d} H) = \omega (\mathcal{X}  _F ,\mathcal{X} _H), 
\]
that is $\langle \mathbf{d} F , \pi ^\sharp \cdot \mathbf{d} H \rangle = \langle \mathbf{d} F, \mathcal{X}  _H \rangle$. On the other hand, by definition, $ \omega (\mathcal{X}_H,v)   = \mathbf{d} H \cdot v $, and so $  \langle \omega ^\flat \cdot  \mathcal{X}  _H , v \rangle = \langle \mathbf{d} H , v\rangle $, whence 
\[
    \mathcal{X}  _H = \omega ^\sharp \cdot \mathbf{d} H 
\]
since $ \omega ^\sharp = (\omega ^\flat) ^{ - 1 } $ (see \cite{abraham_foundations_1978}). Thus $ \pi ^\sharp \cdot  \mathbf{d} H = \omega ^\sharp \cdot \mathbf{d} H $, for all $H$, and thus 
\[\pi  ^\sharp = \omega ^\sharp.\]

\begin{definition}
	A vector field on a symplectic manifold $ ( M , \omega ) $ is called {\bf locally Hamiltonian } if for every $ x \in M $ there is a neighborhood $U$ of $x$ and a smooth function $ H _U $ defined on this neighborhood such that $\mathbf{i} _\mathcal{X} \omega = \mathbf{d} H  _U $, that is $ \mathcal{X} $ is Hamiltonian in $U$ with the locally defined Hamiltonian $ H _U $.
\end{definition}
If the manifold $M$ has zero first group of real homology $H^1(M,\mathbb R)$, then all local Hamiltonian vector fields are globally Hamiltonian \cite{fomenko_integrability_1988}.
%[A. T. Fomenko, Integrability and nonitegrability in geometry and mechanics, Kluwer Academic Publishers, Dordrecht 1988 ].

\begin{proposition}\label{prop:locally_Hamiltonian}
	The following statements are equivalent:
	\begin{enumerate}[(i)]
		\item $\mathcal{X}$ is locally Hamiltonian.
		\item $ \mathbf{d} (\mathbf{i} _\mathcal{X} \omega)=0 $, that is $ \mathbf{i} _\mathcal{X} \omega $  is closed.
		\item $ \mathcal{L} _\mathcal{X} \omega = 0 $.
	\end{enumerate}
\end{proposition}
\begin{proof}
    See \cite{abraham_foundations_1978}.
    %	We show that $ (i) $ iff $ ( ii ) $. $\mathcal{X}$ is locally Hamiltonian iff $ \mathbf{i} _\mathcal{X} \omega $ is locally exact, that is, for each $x$ there is a neighborhood $U$  of $x$  and a locally defined function $ H  _U $ such that $ \mathbf{i} _\mathcal{X} \omega = \mathbf{d}   H _U  $. By the Poincar\'e lemma, this is equivalent to $ \mathbf{d} ( \mathbf{i} _\mathcal{X} \omega ) = 0 $.
%	We show that $ ( ii ) $ iff $ ( iii ) $. By Cartan's magic formula $ \mathcal{L} _\mathcal{X} \omega = \mathbf{d} (\mathbf{i}  _\mathcal{X} \omega) + \mathbf{i} _\mathcal{X} \mathbf{d} \omega = \mathbf{d} ( \mathbf{i} _\mathcal{X} \omega ) $, since $ \omega $ is closed.  Thus $ \mathbf{d} ( \mathbf{i} _\mathcal{X} \omega ) = 0 $ iff $ \mathcal{L} _\mathcal{X} \omega = 0 $.
\end{proof}
\begin{remark}
	Another equivalent way to define locally Hamiltonian vector fields is the following. A vector field is locally Hamiltonian if there exists a closed 1-form $\alpha$  such that $ \mathcal{X} = \omega ^\sharp \cdot  \alpha $. In fact, if $ \mathcal{X} = \omega ^\sharp \cdot \alpha $, then $ \mathbf{i} _\mathcal{X} \omega = \alpha  $. So saying that $ \alpha $ is closed is equivalent to saying that $\mathbf{i} _\mathcal{X} \omega $ is closed.
\end{remark}

\begin{theorem}[Darboux' Theorem for the symplectic case]
	Let $ (M , \omega) $ be a symplectic manifold of dimension $ 2n $, then for each point $ x \in M $ there exists a neighborhood $U$ of  $x$ with coordinates $ (q ^1  , \ldots , q ^n  , p _1 , \ldots , p _n) $ (called {\bf canonical coordinates })  such that
	\[
		\omega|_U  = \sum _{ i = 1 } ^n \mathbf{d} q ^i  \wedge \mathbf{d} p _i.
	\]
\end{theorem}
\begin{proof}
	See \cite{abraham_foundations_1978}.
\end{proof}
\begin{proposition}\label{prop:Hamilton's equations}   Let $ (M , \omega) $ be a symplectic manifold and let $ x \in M $, and let $\mathcal{X} $ be a locally Hamiltonian vector field. Let $U$ be a neighborhood of $x$ with canonical coordinates  
    $(\mathbf{q} , \mathbf{p}) =( q ^1  , \ldots , q ^n  , p _1 , \ldots , p _n) $
, and $ H _U $ such that 
   $ \mathbf{i} _\mathcal{X} \omega  = \mathbf{d} H _U $. In these coordinates
	\[
		\mathcal{X} = \begin{bmatrix}[1.3]
		\frac{ \partial H _U  }{ \partial \mathbf{p}  }  \\
		- \frac{ \partial H _U  }{ \partial \mathbf{q}  }\\
		\end{bmatrix} = \mathbb{J}  \nabla H _U
	\]
	with
	\[
		\mathbb{J}  = \begin{bmatrix}
		\mathbf{0}  & \mathbf{1}  \\
		{\bf - 1 }  & \mathbf{0}
		\end{bmatrix} \mbox{ and }
		\nabla H _U =  \begin{bmatrix}[1.3]
		\frac{ \partial H _U  }{ \partial \mathbf{q}  }  \\
		\frac{ \partial H _U  }{ \partial \mathbf{p}  }\\
		\end{bmatrix}
	\]
	where  $\mathbf{1} $ and $ \mathbf{0} $ define the $ n \times n $   identity and zero matrix, respectively.

	%\[
	%    \mathcal{X} = \left(\frac{ \partial H _U  } { \partial \mathbf{p}  } ,- \frac{ \partial H _U  } { \partial \mathbf{q}  } \right) .
	%\]
	Thus $ (\mathbf{q}  (t) , \mathbf{p}  (t)) $ is an integral curve of $ \mathcal{X} $ if an only if {\bf Hamilton's equations} hold:
	\[
		\dot q ^i  = \frac{ \partial H _U  } { \partial p _i } , \quad \dot p _i = -\frac{ \partial H _U  } { \partial q ^i  }
	\]
	for $ i = 1 , \ldots , n $.
\end{proposition}
\begin{proof}See \cite{abraham_foundations_1978}.

\end{proof}

\subsection{BiHamiltonian structures}
Here we briefly recall some of the most important facts about biHamiltonian structures. 
\begin{definition}
   Let $ \pi _1 $ and $ \pi _2 $ be two Poisson bivector fields defined on a manifold $M$. We say that $ \pi _1 $ and $ \pi _2 $ are {\bf compatible } if their Schouten-Nijenhuis  bracket is zero, that is if $ [\pi _1 , \pi _2 ] = 0 $. The triple $ (M, \pi _1 , \pi _2) $ is called a {\bf biHamiltonian manifold}.  
\end{definition}
We now recall the local coordinate representations of the pull-back of two forms and bivectors, and of the push-forward of vector fields. These expressions are useful in  some of the computations done in the following sections. 
Let  $M$ and $N$ be manifolds, $ f: M \to N $ be a smooth map, and let  $ \rho $ be a two-form on $N$.  Recall that $ f ^\ast \rho   $ denotes the pull-back of $ \rho   $ by  $ f $,  that in local coordinates takes the form
 %then the pull-back $ f ^\ast \rho $ of $ \rho $ by $ f $ in local coordinates is
\[
	(f ^\ast \rho) _{ jk } =  \sum _{ rs }  (\rho _{ rs } \circ f ) \left( \frac{ \partial f ^r } { \partial x ^j } \right) \left( \frac{ \partial f ^s } { \partial x ^k }  \right).
\]
Now suppose $ f $ is a diffeomorphism, and let $ \pi $ be a bivector field on $N$, then  the pull-back of $ \pi  $ by  $ f $ in local coordinates is %Recall that $ f ^\ast \pi  $ denotes the pull-back of $ \pi  $ by  $ f $, that in local coordinates takes the form
\begin{equation}\label{eqn:pull_back_bivector}
	(f  ^\ast \pi) ^{ jk } = \sum _{ rs } \left( \frac{ \partial (f ^{ - 1 }) ^j } { \partial X ^r } \circ f \right) \left(  \frac{ \partial (f ^{ - 1 }) ^k } { \partial X ^s } \circ f \right) \pi ^{ rs } \circ f .
\end{equation}

%Let  $M$ and $N$ be manifolds and $ f: M \to N $ be a smooth map, and let $ \mathcal{X} $ be a vector field on $M$.  Let $( x ^1 , \ldots , x ^{n })  $ and $(X ^1 , \ldots , X ^{ n })  $ be local coordinates on $M$ and $N$, respectively.
Suppose  $\mathcal{X} $ is a vector field on $M$. Then the pushforward  $ f _\ast \mathcal{X} $ of $ \mathcal{X} $ by  $ f $,  in local coordinates, takes the form
\[
	(f _\ast \mathcal{X}) ^j  = \sum _k \left( \frac{ \partial f ^j } { \partial x ^k } \circ f ^{ - 1 } \right) (\mathcal{X} ^k \circ f ^{ - 1 }).
\]
\begin{proposition}\label{prop:diffeo_poisson}
	Let  $M$ be a manifold and let $ ( N , \pi ) $ be a Poisson manifold. Suppose $f:M \to N$ is a diffeomorphism. Then  $f ^\ast  \pi  $ is a Poisson tensor on M.
\end{proposition}
\begin{proof}
	Since the  bracket has the following property (see \cite{vaisman_lectures_1994}):
	\[
		f ^\ast [\pi , \pi ] = [ f ^\ast  \pi , f ^\ast \pi ].
	\]
	it follows that  $ [\pi , \pi ] = 0 $ implies $ [ f ^\ast  \pi , f ^\ast \pi  ]=0 $.
\end{proof}

\begin{corollary}  \label{cor:compatible_poisson}  
  Let 	 $M$ be a manifold and let $ ( N , \pi _1 , \pi _2  ) $ be a biHamiltonian manifold, that is $ \pi _1 $ and $ \pi _2 $ are compatible. Suppose $f:M \to N$ is a diffeomorphism. Then   $ f ^\ast \pi _1 $ and $ f ^\ast \pi _2 $ are compatible and thus $ (M , f ^\ast \pi _1 ,  f^\ast \pi _2) $ is a biHamiltonian manifold.   
\end{corollary} 
\begin{proof} 
    Since
    \[
        [ \pi _1 + \pi _2 , \pi _1 + \pi _2 ] = [ \pi _1 , \pi _1 ] + [ \pi _2 , \pi _2 ] + 2 [ \pi _1 , \pi _2 ] 
    \]
   we have 
   \[
      f ^\ast  [ \pi _1 + \pi _2 , \pi _1 + \pi _2 ] = f ^\ast  [ \pi _1 , \pi _1 ] + f ^\ast  [ \pi _2 , \pi _2 ] + 2 f ^\ast  [ \pi _1 , \pi _2 ] 
   \] 
   and,
   \[
        [ f ^\ast ( \pi _1 + \pi _2)  , f ^\ast (\pi _1 + \pi _2 )  ] = [ f ^\ast\pi _1 , f ^\ast \pi _1 ] + [ f ^\ast  \pi _2 , f ^\ast  \pi _2 ] + 2 [ f ^\ast \pi _1 , f ^\ast \pi _2 ]. 
    \]
    Comparing the two equations above,  by Proposition \ref{prop:diffeo_poisson} we obtain that $ f ^\ast [\pi _1 , \pi _2 ]= [ f ^\ast \pi _1 , f ^\ast \pi _2 ] $.  
\end{proof} 

Corresponding to the bivector fields $ \pi _1 $ and $ \pi _2 $ we can define the  Poisson brackets $\{ F , G \} _1 = \pi _1 (\mathbf{d} F , \mathbf{d} G) $ and $ \{ F , G \} _2 = \pi _2 (\mathbf{d} F , \mathbf{d} G) $. With these notations we give the following
\begin{definition} 
Let $ (M , \pi _1 , \pi _2) $ be a biHamiltonian manifold and suppose there exist functions $ H _1 $ and $ H _2 $ on $M$ for which 
\[
    \mathcal{X} [F] = \{ F , H _1 \} _1 = \{ F , H _2 \} _2
\]
for every function $F$ on $M$. 
Then $ \mathcal{X} $ is called a {\bf biHamiltonian vector field}. 
\end{definition}

The importance of biHamiltonian structures lies in the fact that, in certain situations, they can be used to show complete integrability. We do not  give a complete account, but the main idea is that one can  use them to construct a set of first integrals in involution by constructing a biHamiltonian hierarchy \cite{magri_eight_2004,laurent-gengoux_poisson_2012,bolsinov_compatible_1992}
\begin{definition}
   Let $ (M , \pi _1 , \pi _2) $ be a biHamiltonian manifold. A {\bf biHamiltonian  hierarchy } on $M$ is a sequence of functions $ \{ F _i \}_{ i \in \mathbb{Z}  } $ such that 
   \[
       \{ \cdot ,F _{ i + i } \}_1  = \{ \cdot ,F _i \} _2 
   \]   
   for every $ i \in \mathbb{Z}  $. 
\end{definition}
The following lemma explains why a biHamiltonian hierarchy yields  functions in involution.
\begin{proposition} 
Suppose $ \{ F _i \}_{ i \in \mathbb{Z}  } $  is a biHamiltonian hierarchy, then $ \{ F _i , F _j \} _1  = \{ F _i , F _j \} _2  =  0$  for all $ i<j \in \mathbb{Z}  $.
\end{proposition}
\begin{proof}
    \begin{align*}
        \{ F _i , F _j \} _1 & = \{ F _i , F _{ j - 1 }\}_2  \\
        & = \{ F _{ i + 1 } , F _{ j - 1 } \} _1 \\
        & = \ldots \\
        & = \{ F _j , F _i \} _1,  
    \end{align*}  
    so that $ \{ F _i , F _j \} _1 = 0 $ by skew-symmetry. Hence, the $ F _i $'s are in involution with respect to the Poisson bracket $ \{ \cdot, \cdot \} _1 $, and also with respect to $ \{ \cdot, \cdot \} _2$ , since   $ \{ F _i , F _j \} _2 = \{ F _i , F _{ j + 1 } \} _1 $.

\end{proof}

\section{Canonoid Transformations }

\begin{definition}
	Let  $ ( M , \omega ) $ be a symplectic manifold, and let $ \mathcal{X} $ be a locally Hamiltonian vector field on $M$, that is, for each $ x \in M $, there is a neighborhood $U$ of $x$ and  locally defined function $  H_U $ such that $ \mathbf{i} _\mathcal{X} \omega = \mathbf{d}  H_U  $. A diffeomorphism %transformation
	$ f: M \to M $ is said to be {\bf canonoid} with respect to the vector field $ \mathcal{X} $ if the transformed vector field $ f _\ast \mathcal{X}$ is also locally Hamiltonian, that is, for each $ x \in M $, there is a neighborhood $V$ of $ f (x) $ and a locally defined function $  K_V $ such that $ \mathbf{i} _{ f _\ast \mathcal{X} } \omega =  \mathbf{d} K_V$.
\end{definition}
This is equivalent to saying that, for each $y \in M $, there is a neighborhood $V$ of $y$ and a locally defined function $ K _V $ such that $  \mathbf{i}  _\mathcal{X} (f ^\ast \omega) = f ^\ast \mathbf{d} K _V $. This is also equivalent to  $ \mathcal{L} _\mathcal{X} (f ^\ast \omega) = 0 $.

\begin{remark}
	By \ref{prop:Hamilton's equations} the previous definition means that, for each point $x \in M$ the system of equations associated with $ \mathbf{i} _\mathcal{X} \omega = \mathbf{d} H _U $  can be written, in Darboux coordinates on the neighborhood  $U$ of $x$, as:
	\begin{equation}\label{eqn:Hamilton_equations_H}
		\dot q ^i  = \frac{ \partial H _U  } { \partial p _i } , \quad \dot p _i = -\frac{ \partial H _U  } { \partial q ^i  }
	\end{equation}
	and the system associated with $ \mathbf{i } _{ f _\ast \mathcal{X} } \omega = \mathbf{d} K _V $ can be written, in Darboux coordinates on the neighborhood  $V$ of $f (x) $, as:
	\begin{equation} \label{eqn:Hamilton_equations_K}
		\dot Q ^i  = \frac{ \partial K _V  } { \partial P _i } , \quad \dot P _i =- \frac{ \partial K _V  } { \partial Q ^i  }.
	\end{equation}
	so that the transformation $f$ carries the system of Hamilton's equations \ref{eqn:Hamilton_equations_H} again into a system of Hamilton's equations \ref{eqn:Hamilton_equations_K}. 
	
%	In these coordinates, as well as in any other coordinates in $U$ and $V$, the transformation $f$ is determined by equations $P_i=P_i(q^j,p_j)$, $Q^i=Q^i(q^j,p_j)$. If, as in our case, $f$ is a diffeomorphism, then the transformation coincides locally  with the coordinate change $(P_i,Q^i) \leftrightarrow (p_j,q^j)$ and $f^*A$, $f_*B$, for any suitable object $A$ or $B$, are locally determined by writing $A$ in coordinates $(p_j,q^j)$ and $B$ in $(P_i,Q^i)$.
\end{remark}

We modify an example found in \cite{fasano_analytical_2006} to give a general framework to construct linear canonoid transformations. Let us consider Hamiltonian systems on the symplectic manifold $ (M , \omega) = (\mathbb{R}^{2n}, \omega)   $, let $ x = (\mathbf{q} , \mathbf{p}) $ be Darboux coordinates on $ \mathbb{R}  ^{ 2n } $, then the symplectic form can be written as $ \omega = \sum \mathbf{d} q ^i  \wedge \mathbf{d} p _i $. In this case a diffeomorphism $ f : \mathbb{R}  ^{ 2n } \to \mathbb{R}  ^{ 2n } $ is called a  {\bf canonical  transformation} if $f$ preserves the 2-form $ \omega = \sum \mathbf{d} q ^i  \wedge \mathbf{d} p _i $, that is if $ f ^\ast \omega = \omega $. A diffeomorphism is canonical if and only if the matrix  representation of  $ \mathbf{d} f $ in the canonical basis of $ \mathbb{R}^{2n}$, namely $ [\mathbf{d} f]$  , is a symplectic matrix, that is $ [\mathbf{d} f] ^t \mathbb{J}  [\mathbf{d} f] = \mathbb{J}  $.

Any quadratic Hamiltonian can be written as
\[
	H ( x ) = \frac{1}{2} x ^t S x,
\]
where $S$ is a real symmetric constant $ 2n \times 2n $ matrix. With these notations, the Hamiltonian vector field corresponding to $ H $ can be written as $ \mathbb{J}  S x $, and Hamilton's equations take the form
\[
	\dot x =  \mathbb{J}  S x.
\]
Hamilton's equations above define a linear Hamiltonian system with constant coefficients.
Consider the  transformation $ f: \mathbb{R} ^{ 2n } \to \mathbb{R}^{2n }$, defined by  $ X =f(x) =  A x $, with $A$
an invertible matrix. Then  the vector field $ \mathcal{X} $  is transformed to $ f _\ast \mathcal{X}  =  A \mathbb{J}  S A ^{ - 1 } X $  and the system of Hamilton's equations is transformed into a new system of $ 2n $ differential equations
\[
	\dot X = A \mathbb{J}  S A ^{ - 1 } X
\]
expressed in terms of the variables $ X = ( \mathbf{Q}  , \mathbf{P}  ) $. In general, the new system does not have the canonical structure, that is, it is not necessarily true that there exists a Hamiltonian $ K ( X) $ such that

\[\dot X = A \mathbb{J}  S A ^{ - 1 }X = \mathbb{J}  \nabla _{ X } K\]
and thus not every transformation of this type is canonoid.
However, it is easy to see that, in order to preserve the canonical structure, we must have
\[
	A \mathbb{J}  S A ^{ - 1 } = \mathbb{J}  C
\]
for some symmetric matrix $C$. We can rewrite this condition as $ A ^t \mathbb{J}  A \mathbb{J}  S = - A ^t C A $. It follows that the existence of a symmetric matrix $C$ is equivalent to the symmetry condition
\[
	A ^t \mathbb{J}  A \mathbb{J}  S = S \mathbb{J}  A ^t \mathbb{J}  A
\]
or
\begin{equation}\label{eqn:cancond}
	\Gamma ^t \mathbb{J}  S + S \mathbb{J}  \Gamma = 0
\end{equation}
with $ \Gamma = A ^t \mathbb{J}  A = - \Gamma ^t$. Thus, in this case, the condition for having a canonoid transformation reduces to equation \eqref{eqn:cancond}.
\begin{remark}
	If the transformation is canonical the matrix $A$ is symplectic ($ A ^T \mathbb{J}  A = \mathbb{J}$),  and thus, $ \Gamma = \mathbb{J}  $ and the condition is satisfied. The same is true if $ \Gamma = a \mathbb{J}  $ ( with $ a \neq 0 $).
\end{remark}
 
\begin{remark}
	If $A$ represents a  rescaling of the given coordinates, namely
	\begin{eqnarray*}
		A_{ii}&=& a_i, \, 1\leq i \leq n, \\
		A_{ii}&=& b_i,\, n< i \leq 2n,\\
		A_{ij}&=& 0, \, i\neq j,
	\end{eqnarray*}
	then, $\Gamma \mathbb J=\mathbb J \Gamma = B$, a diagonal matrix determined by
	\begin{eqnarray*}
		B_{ii}&=& -a_ib_i, \, 1\leq i \leq n, \\
		B_{ii}&=& -a_ib_i,\, n< i \leq 2n,\\
		B_{ij}&=& 0, \, i\neq j.
	\end{eqnarray*}
	The transformation is canonoid if and only if (\ref{eqn:cancond}) holds, i.e.
	\[
		BS=SB.
	\]
	When the rescaling is a point-transformation, then $a_ib_i=1$ and it is always canonoid.
	
\end{remark}
We now find more explicit conditions to have canonoid transformations. Write  the matrices $\Gamma $ and $ S $ in  terms of $ n \times n $ blocks
as follows:
\[
	\Gamma = \begin{bmatrix}
	\lambda  & \mu  \\
	- \mu ^t  & \nu
	\end{bmatrix} , \quad \quad
	S = \begin{bmatrix}
	\alpha  & \beta  \\
	\beta  ^t  &\gamma
	\end{bmatrix}
\]
where $ \lambda ^t = - \lambda $, $ \nu ^t = - \nu $, $ \alpha ^t = \alpha $, and $ \gamma ^t = \gamma $.
The equations $ \Gamma \mathbb{J}  S = S \mathbb{J}  \Gamma $ leads to the system
\begin{align*}
	- \lambda \beta ^t + \mu \alpha & = \alpha \mu ^t + \beta \lambda    \\
	- \lambda \gamma + \mu \beta    & = - \alpha \nu + \beta \mu         \\
	\mu ^t \beta ^t + \nu \alpha    & = \beta ^t \mu ^t + \gamma \lambda \\
	\mu ^t \gamma + \nu \beta       & = - \beta ^t \nu + \gamma \mu.
\end{align*}
If we let
\[
	A = \begin{bmatrix}
	a & b \\
	c & d
	\end{bmatrix}
\]
then
\begin{align}\begin{split} \label{eqn:atja }
	\Gamma            & =  \begin{bmatrix}
	a ^t              & c ^t               \\
	b ^t              & d^t
	\end{bmatrix}
	\begin{bmatrix}
	\mathbf{0}        & \mathbf{1}         \\
	{\bf - 1 }        & \mathbf{0}
	\end{bmatrix}
	\begin{bmatrix}
	a                 & b                  \\
	c                 & d
	\end{bmatrix}\\
	                  & =
	\begin{bmatrix}
	-c ^t a + a ^t c  & - c ^t b + a ^t d  \\
	- d ^t a + b ^t c & - d ^t b + b ^t d
	\end{bmatrix}=
	\begin{bmatrix}
	\lambda           & \mu                \\
	- \mu ^t          & \nu
	\end{bmatrix}
	\end{split}
\end{align}

\begin{proposition}\label{can}
	Given the linear Hamiltonian system of Hamiltonian $ H ( x ) = \frac{1}{2} x ^t S x $, the invertible transformation $ X = A x $ preserves the canonical structure of Hamilton's equations for all Hamiltonians if and only if $\Gamma = a \mathbb{J}  $ for some constant $a \neq 0 $.
\end{proposition}
\begin{proof}
	Showing that $\Gamma = a \mathbb{J}  $ for some constant $a \neq 0 $ satisfies the conditions is trivial
	Conversely, consider  the particular case $ \alpha = \gamma = \mathbf{0} $. We find that $\mu$ must commute with every $ n \times n $ matrix, and therefore $ \mu = a \mathbf{1} $. Choosing $ \alpha = \beta = \mathbf{0}  $ we find $ \lambda = \mathbf{0}  $. From $ \beta = \gamma = \mathbf{0} $ it follows that $ \nu = \mathbf{0} $. Hence $ \Gamma = a \mathbb{J}  $, and in addition, from $ A ^t \mathbb{J}  A = a \mathbb{J}  $ it follows that $ \mathbb{J}  A \mathbb{J}  = - a ( A ^{ - 1 } ) ^t $. We finally find that $C = a ( A ^{ - 1 } ) ^t S A ^{ - 1 } $ and the new Hamiltonian is $ K ( X ) = \frac{1}{2} X ^t C X $. If $A$ is symplectic it holds that $ K ( X ) = H ( x ) $, and if $ a \neq 1 $ we find $ K ( X ) = a H ( x ) $.
\end{proof}
\begin{example}
	Let $S$ be  such that $ \beta = \alpha  = \mathbf{0}  $, and $ \gamma = \mathbf{1} $. Then the equations reduce to
	\begin{align*}
		\mathbf{0}       & = \mathbf{0}                        \\
		- \lambda \gamma & =- \lambda \mathbf{1} =  \mathbf{0} \\
		\gamma \lambda   & = \mathbf{1} \lambda =  \mathbf{0}  \\
		\mu ^t \gamma    & = \gamma  \mu
	\end{align*}
	Hence, from the first (or second) equation $ \lambda = \mathbf{0} $, from the last equation $\mu ^t = \mu $. Hence, by equation ( \ref{eqn:atja }) the only requirements are that $  c ^t a =  a ^t c  $ (i.e. $c ^t a $ is symmetric), and that $ - c ^t b + a ^t d = d ^t a - b ^t c $ (i.e. $ - c ^t b + a ^t d $ is symmetric ).
\end{example}

\begin{example}\label{ex:free_particle}
	We specialize the previous example. Let $ H = \frac{1}{2} ( p _1 ^2 + p _2 ^2 )$ then $ S $ is a $ 4 \times 4 $ matrix with $\beta = \alpha = \mathbf{0} $ and  $\gamma = \mathbf{1} $. Suppose that $ a = \mathbf{1} $, $b = c = \mathbf{0} $, and that
	\[
		d ^{ - 1 }  = \begin{bmatrix}
		m & l \\
		l & n
		\end{bmatrix}
	\]
	Clearly $ - c ^t b + a ^t d = a ^t d = d $ is symmetric, since $d$ is symmetric. Moreover, $ c ^t a = \mathbf{0} $  and so it is symmetric. Therefore, this transformation satisfies the conditions obtained in the previous example.
	Then we can compute $C$ as $ C = - \mathbb{J}  A \mathbb{J}  S A ^{ - 1 } $. We obtain
	\[
		C = \begin{bmatrix}
		0 & 0 & 0 &0 \\
		0 & 0 &0 &0\\
		0 &0 & m &l \\
		0& 0 &l &n \\
		\end{bmatrix}
	\]
	and hence $ K = \frac{1}{2} ( m P _1 ^2 + n P _2 ^2+ 2 l P _1 P _2 ) $.
\end{example}
We now show that given a canonoid transformation it is possible to find an additional symplectic structure and an additional first integral, and thus canonoid transformations can be used to find bihamiltonian structures and to study the integrability of Hamiltonian systems.
 
Let $ \omega _1 = \omega   $ be the symplectic form  defined by
\[
	\omega  ( x , y ) = x ^t \mathbb{J}  y.
\]
Then, the Hamiltonian vector field with Hamiltonian $H$ satisfy the equation
\[
	\omega _1 ( \mathcal{X}  _H , v ) = \mathbf{d} H(x) \cdot v
\]
for all $v \in \mathbb{R} ^{ 2n }$.
Similarly, let $ \Omega $ be the symplectic form in the ``transformed space". $\Omega$ is defined as follows:
\[
	\Omega ( X , Y ) = X ^t \mathbb{J}  Y.
\]
Then, the Hamiltonian vector field with Hamiltonian $K$ is  satisfies the equation
\[
	\Omega ( \mathcal{X}  _K , v ) = \mathbf{d} K \cdot v
\]
for all $v \in \mathbb{R} ^{ 2n }$, where $ \mathcal{X} _K = f _\ast (\mathcal{X} _H) $. %  $ \mathcal{X}   _K = T ( \mathcal{X}  _H ) = A \mathcal{X}  _H $.
Let $f$ be the linear transformation defined as $ X = f ( x ) = A x $, where $A$ is the $ 2n \times 2n $ invertible matrix introduced above.  We can use $f$ to define a new  canonical form in the ``$x $" space  by pulling back the canonical form $\Omega$:
\begin{align*}
	\omega _2 ( x, y ) & = (f ^\ast  \Omega ) ( x, y ) = \Omega ( f ( x ) , f(y) ) \\
	                   & = \Omega(Ax, Ay ) = (Ax) ^t\mathbb{J}  (Ay)               \\
	                   & = x ^t ( A ^t \mathbb{J}  A) y
\end{align*}
which gives an explicit expression of the symplectic form $ \omega _2 $ in terms of the matrix $A$.

Then we can write, by pulling back the equation $  \Omega ( \mathcal{X}  _K , v ) = \mathbf{d} K \cdot v $
\[
	\omega _2(\mathcal{X}  _H ,v) =  (f ^\ast  \Omega ) ( \mathcal{X}  _H , v ) = f^\ast (\mathbf{d} K)(v)= \mathbf{d} (f ^\ast K)(v)
\]
where $ (f ^\ast K)(x) = K ( f ( x ) )= K ( A x )  $. If we introduce $ H _2( x ) =  (f ^\ast K)(x) $ we can write
\[
	\omega _2(\mathcal{X}  _H ,v) = \mathbf{d} H _2 \cdot v
\]
Hence, the vector field $ \mathcal{X}  _H $ is Hamiltonian  with respect to the symplectic form $ \omega _1 $ and also Hamiltonian with respect to the symplectic form $\omega _2 $.

\begin{example} We now continue example \ref{ex:free_particle}. We compute $\omega _2 $ and $ H _2 $ for this example. The transformation is given by the matrix
	\[ A = \begin{bmatrix}
		\mathbf{1}  & \mathbf{0}  \\
		\mathbf{0}  & d
		\end{bmatrix}
	\]
	where $\mathbf{1} $ and $ \mathbf{0} $ are the $ 2 \times 2 $ identity matrix and zero matrix, respectively. The matrix $d$ is given by
	\[
		d=  \frac{ 1 } {m n - l ^2 } \begin{bmatrix}
		n & - l \\
		- l  & m
		\end{bmatrix}
	\]
	Then the matrix representative of $\omega _2 $ is given by
	\[
		[\omega _2 ]=  A ^t \mathbb{J}  A = \begin{bmatrix}
		\mathbf{0}  & d \\
		- d & \mathbf{0}
		\end{bmatrix},
	\]
	note that the matrix $A$ is symplectic if and only if $l = 0 $ and $ m= n = 1 $, so that this transformation is symplectic if and only if it is the identity.
	The new Hamiltonian, obtained after some computations, is
	\[
		H _2 ( x ) = K ( A x ) = \frac{ 1 } { 2(m n - l ^2) }[n p _1 ^2 + m p _2 ^2 - 2l p _1 p _2 ].
	\]
	Clearly $ H_2 $ is a  first integral of the system with Hamiltonian $ H $, since $ \{ H , H _2 \} = 0 $.
\end{example}
\begin{example}
	We now consider a more interesting example, namely the harmonic oscillator. In this case $ \beta = \mathbf{0} $, and $ \alpha = \gamma = \mathbf{1} $. The conditions for having a canonoid transformation reduce to  $ \nu = \lambda $ , and $ \mu = \mu ^t $ (i.e. $ \mu $ is symmetric). Now suppose $S$ is a $ 2 \times 2 $ matrix.
	\begin{enumerate}[(a)]
		\item We can specialize the previous transformation by taking $a =d= \begin{bmatrix}
		      1 & 1 \\
		      1 & - 1
		\end{bmatrix} $, $ b = \begin{bmatrix}
		2 & 0 \\
		0 & 1
		\end{bmatrix} $, and $ c = \begin{bmatrix}
		1 & 0 \\
		0 & 2
		\end{bmatrix} $. Then
		\[\Gamma = \begin{bmatrix}
			0 & 1 & 0 & 0 \\
			- 1 & 0 & 0 & 0 \\
			0 & 0 & 0 & 1 \\
			0 & 0 & - 1 & 0
			\end{bmatrix},
		\]
		where $ E _3 = \Gamma $ is the matrix of the new symplectic form. Moreover,
		\[C =
			\begin{bmatrix}
				-2 & 3  & 0  & 3  \\
				3  & 4  & -6 & 0  \\
				0  & -6 & 4  & -3 \\
				3  & 0  & -3 & -2
			\end{bmatrix}
		\]
		so that $ K = \frac{1}{2} X ^t C X $ is the Hamiltonian of the transformed system. Transforming $K$ to the old coordinates yields the Hamiltonian
		\[ W _1 = (q _2 p _1 - q _1 p _2 ).
		\]
		\item We can also specialize the previous transformation by taking $ b = c = 0 $ and taking $a$ and $d$ to be symmetric matrices, then
		      \[
		      	\Gamma  = \begin{bmatrix}
		      	0 & a ^t d \\
		      	- d ^t a & 0
		      	\end{bmatrix}, \quad
		      	a  = \begin{bmatrix}
		      	a _{ 11 }  &  a _{ 12 }  \\
		      	a _{ 12 }  & a _{ 22 }
		      	\end{bmatrix}, \quad
		      	d   = \begin{bmatrix}
		      	d _{ 11 }  & d _{ 12 }  \\
		      	d _{12 }  & d _{ 22 }
		      	\end{bmatrix}
		      \]
		      and $ \mu = a ^t d $ is a symmetric matrix. Then
		      \[
		      	C = \begin{bmatrix}
		      	C _1  & \mathbf{0}  \\
		      	\mathbf{0}  & C _2
		      	\end{bmatrix}
		      \]
		      where
		      \[C _1 = \frac{ 1 } { \det a } \left [\begin{smallmatrix}
		      		a _{ 22 } d _{ 11 } - a _{ 12 } d _{ 12 }  &  a _{ 11 } d _{ 12 }   - a _{ 12 } d _{ 11 }  \\
		      		a _{ 22 } d _{ 12 } - a _{ 12 } d _{ 22 }  & a _{ 11 } d _{ 22 } -  a _{ 12 } d _{ 12 }
		      	\end{smallmatrix}\right], \quad
		      	C _2 = \frac{ 1 } { \det d }\left [\begin{smallmatrix}
		      		a _{ 11 } d _{ 22 }  - a _{ 12 } d _{ 12 } &   a _{ 12 } d _{ 11 } - a _{ 11 } d _{ 12 } \\
		      		a _{ 12 } d _{ 22 }  - a _{ 22 } d _{ 12 }  &   a _{ 22 } d_{ 11 } - a _{ 12 } d _{ 12 }
		      	\end{smallmatrix}  \right]
		      \]
		      then $ K = \frac{1}{2} X ^t C X $, and $ H _2 = \frac{1}{2} x ^t (A ^t C A) x $, where
		      \[
		      	A ^t C A = \left[\begin{smallmatrix}
		      		a_{11} d_{11} + a_{12} d_{12} & a_{11} d_{12} + a_{12} d_{22} & 0 & 0 \\
		      		a_{12} d_{11} + a_{22} d_{12} & a_{12} d_{12} + a_{22} d_{22} & 0 & 0 \\
		      		0 & 0 & a_{11} d_{11} + a_{12} d_{12} & a_{12} d_{11} + a_{22} d_{12} \\
		      		0 & 0 & a_{11} d_{12} + a_{12} d_{22} & a_{12} d_{12} + a_{22} d_{22}
		      	\end{smallmatrix}\right]
		      \]
		      If, in particular,  we set $ a _{ 11 } = d _{ 12 } = a _{ 22 } =  0 $, $ a _{ 12 } = 1$, $ d _{ 11 } = 1 $, and $ d _{ 22 } =  1 $,  then  we obtain $ W _2 =  (q _1 q _2 + p _1 p _2)  $, which is a first integral, and the corresponding symplectic form has the following matrix representation
		      \[
		      	E _2 =  \begin{bmatrix}
		      	0 & 0&0&  1 \\
		      	0 & 0& 1&0\\
		      	0&-1&0&0\\
		      	-1&0 & 0 & 0
		      	\end{bmatrix}
		      \]
	
		      If, instead,  we set $ a _{ 12 } = d _{ 12 } = 0 $, $ a _{ 11 } = a _{ 22 } = d _{ 11 } =1 $, and  $  d _{ 22 } = -1 $ then  we obtain $ W _3 = \frac{1}{2} ( q _1 ^2 + p _1 ^2 - q _2 ^2 - p _2 ^2) $, which is a first integral, and the corresponding symplectic has the following matrix representation
		      \[
		      	E _3 = \begin{bmatrix}
		      	0 & 0&1& 0 \\
		      	0 & 0&0&-1\\
		      	-1&0&0&0\\
		      	0&1 & 0 & 0
		      	\end{bmatrix}
		      \]
	\end{enumerate}
	Now, let $ W _4 = \frac{1}{2} ( p _1 ^2 + p _2 ^2 + q _1 ^2 + q _2 ^2 ) $ and $ E _4 = \mathbb{J}  $. It is easy to see that
	\[W _4 ^2 = W _1 ^2 + W _2 ^2 + W _3 ^2. \]
	Suppose $\mathfrak{ u} ( 2 ) $ is the Lie algebra of the Lie group $ U ( 2 ) $ of $ 2 \times 2 $ unitary matrices. If we consider $ \mathfrak{ u }  ( 2 ) $ as a subspace of $ sp ( 4, \mathbb{R}  )$ , then $ \{ E _1 , E _2 , E _3 , E _4 \} $ is a basis for $ \mathfrak{ u} (2) $. Moreover, the  functions $ W _1, W _2 , W _3$, and  $ W _4 $ are a basis for the vector spaces of all quadratic integrals of the harmonic oscillator vector field.   The map
	\[
		\mathcal{H} : \mathbb{R} ^4 \to \mathbb{R}  ^4 : (q,p)\to ( w _1 ( q , p ) , w _2 ( q,p ) , w _3 ( q,p ) , w _4 ( q , p ) )
	\]
	where $ w _1 = 2 W _1, ~ w _2 = 2W _2 , ~ w _3 =2 W _3 $, and $ w_4 = 2 W _4 $ is called the {\bf Hopf map}.
\end{example}
\section{Linear Canonoid Transformations and the harmonic oscillator }

When we perform a canonoid transformation of some Hamiltonian system,  the integrability or superintegrability  of the system are  preserved. Indeed, a canonoid transformation is essentially a change of coordinates and, consequently, the existence of intrinsic structures like foliations made of invariant tori (Liouville or complete integrability) or the closure of the finite orbits (maximal superintegrability) are left unchanged. A canonoid transformation may only make these structures more or less evident and easy to handle by allowing the determination, together with the new coordinates, of a new Hamiltonian function and a new symplectic structure for the same dynamical system.   %However, quadratic integrability or superintegrability, i.e., in dimension two, the existence of respectively two or three functionally independent quadratic in the momenta constants of the motion, could be affected by such transformation and the study of its behaviour is the aim of this section, in the particular case of linear transformations of the isotropic harmonic oscillator and of a strictly related system. We recall that the two-dimensional isotropic harmonic oscillator is quadratically superintegrable, and functionally independent first integrals are, for example, $W_4$, $W_1$, $W_2$ given in the previous section (see also   \cite{kalnins_completeness_2001}).
	From Prop. \ref{prop:Poissonoid_tr_preserve_constant_of_motion}, it is clear that canonoid and  Poissonoid  transformations of a Hamiltonian system preserve the  functionally independent constants of the motion of the system. Therefore, the transformed of a superintegrable system is again  superintegrable with the transformed constants of motion. If the transformation is linear, then the degree of the polynomial constants of the motion is also preserved by the transformation. We see below how the two-dimensional harmonic oscillator, that admits three quadratic in the momenta  first integrals, $W_2$, $W_3$ and $W_4$ seen above, behaves under linear canonoid transformations.

For our purpose, it is useful that the canonoid transformation of the two dimensional harmonic oscillator leads to a system with  Hamiltonian in either one of the forms
\[
	K_1=\frac 12(P_1^2+ P_2^2+V(Q^1,Q^2)),\quad  K_2=P_1 P_2+V(Q^1,Q^2).
\]
We remark that the manifold where  $K_1$ is defined as the Euclidean plane, while $K_2$  is defined in the Minkowski plane. Moreover, in the Euclidean case the form of the Hamiltonian is non-restrictive, since  any non-degenerate Hamiltonian can be put in  this form by a real  canonical point-transformation. A linear canonical point-transformation changes $K_2$ into $P_1^2-P_2^2+V$. The Hamiltonians  $W_2$ and $W_3$ are then recovered by $K_2$. Under the same constraints, we apply linear canonoid transformations also to the system of Hamiltonian
\[
	H=\frac 12(p_1^2+p_2^2+q_1^2),
\]
that we can consider as an embedding of the one-dimensional harmonic oscillator in $\mathbb E^2$. This system possesses two evident quadratic first integrals, plus a functionally independent third-one
\[
	q_2-p_2\arctan \left( \frac {q_1}{p_1}\right),
\]
not globally defined.

We recall that the linear transformation $A$ is canonoid if and only if (\ref{eqn:cancond}) holds. In order to obtain Hamiltonians of the prescribed form, we must constrain the $2\times 2$ submatrix in the lower-right corner of $ C = - \mathbb{J}  A \mathbb{J}  S A ^{ - 1 } $ to be
\[
	\begin{bmatrix}
		1 & 0 \\
		0 & 1
	\end{bmatrix}, \quad \,\,\,\displaystyle{or} \quad \,\,\,
	\begin{bmatrix}
		0 & 1 \\
		1 & 0
	\end{bmatrix},
\]
respectively, where the matrix $S$ is determined by the  original Hamiltonian. Moreover, the $2\times 2$ submatrices in the upper-right and lower-left corners of $C$ must be equal to zero. We can check if the transformations are canonical thanks to Proposition \ref{can}.
%\subsection{}

With these constraints, we search first for canonoid transformations of the isotropic harmonic oscillator $H=\frac 12 (p_1^2+p_2^2+q_1^2+q_2^2)$. In this case
 
\[
	S = \begin{bmatrix}
	1 & 0 & 0 &0 \\
	0 & 1 &0 &0\\
	0 &0 & 1 &0 \\
	0& 0 &0 &1 \\
	\end{bmatrix}.
\]
 
In order to perform computations with a computer-algebra software, we put $c_{12}=b_{21}=0$. With this simplification we find that only one  matrix $C$  determines a Hamiltonian of the form $K_1$ and  can be obtained by several different matrices $A$. The Hamiltonian $K_1$ is
\[
	K_1=\frac 12(P_1^2+P_2^2+Q_1^2+Q_2^2).
\]
Therefore, the isotropic harmonic oscillator corresponds only to itself under a canonoid transformation of the prescribed type. We remark in particular that it is impossible to obtain by this way anisotropic harmonic oscillators, even if they are  superintegrable when the ratio of the parameters is a rational number. This because one of the constants of the motion must be of degree higher than two in momenta or coordinates \cite{jauch_problem_1940,kalnins_completeness_2001}.
%\subsection{}

If we start from the system of Hamiltonian
\begin{equation}\label{emb}
	H=\frac 12 (p_1^2+p_2^2+q_1^2),
\end{equation}
then
\[
	S = \begin{bmatrix}
	1 & 0 & 0 &0 \\
	0 & 0 &0 &0\\
	0 &0 & 1 &0 \\
	0& 0 &0 &1 \\
	\end{bmatrix}.
\]
We find  that, under the same assumptions regarding $A$ and $C$, the canonoid transformations maps (\ref{emb}) into either
\[
	K_1=\frac 12(P_1^2+P_2^2+Q_1^2),
\]
a Hamiltonian identical to (\ref{emb}), or
\[
	K_1=\frac 12(P_1^2+P_2^2)+k(\alpha_1Q_1-\alpha_2Q_2)^2,
\]
where $k$, $\alpha_1$ and $\alpha_2$ are  constants. In the last case,  the point-transformation
\[
	Q_1=(\alpha_2Y-\alpha_1X)\sqrt{(\alpha_1^2+\alpha_2^2)},\quad Q_2=(\alpha_1Y+\alpha_2X)\sqrt{(\alpha_1^2+\alpha_2^2)},
\]
makes   $K_1$ in the pristine form
\[
	\frac 12(P_X^2+P_Y^2+k'X^2),
\]
for some suitable constant $k'$.  
 
As an example of the second type of canonoid transformations, when $c_{21}=c_{22}=0$, $b_{12}d_{22}=d_{12}b_{22}$, $d_{21}a_{22}=a_{21}d_{22}$, $b_{11}=-(2a_{21}^2a_{22}d_{12}^2d_{22}^2+a_{22}d_{12}^3a_{21}d_{11}d_{22}-d_{12}^2a_{11}a_{22}d_{11}d_{22}^2+d_{12}^3a_{11}a_{21}d_{22}^2+a_{22}d_{12}^4a_{21}^2-d_{12}^4a_{21}^2d_{22}+a_{21}^2a_{22}d_{22}^4) / (a_{22}d_{22}^2d_{12}^2c_{11})$ we have
\[
	K_1=\frac 12(P_1^2+P_2^2)+(d_{22}^2+d_{12}^2)(d_{22}Q_1-d_{12}Q_2)^2.
\]
It can be checked that none of the transformations leading to the last form of $K_1$ is canonical.
%\subsection{}

By imposing the constraint corresponding to $K_2$, we are  mapping the Euclidean  harmonic oscillator into the Minkowski plane. Let us apply first the canonoid transformation to the isotropic oscillator. We obtain, with the same assumptions on $A$, that, for the admissible  solutions, $K_2$ is always in the form
\[
	K_2= P_1P_2+kQ_1Q_2,
\]
for some constant $k$. This is essentially the form of $W_2$ of the previous section. For example, if $a_{21}=c_{11}=b_{22}=d_{11}=0$, $a_{12}d_{12}=a_{22}d_{22}$, $b_{12}d_{12}=-a_{22}c_{22}$, $a_{11}=(a_{22}^2d_{21}^2+c_{21}^2a_{22}^2+c_{21}b_{11}d_{12}^2)/(d_{12}^2d_{21})$, then $k=d_{12}^2/a_{12}^2$, where we assume  that $A$ is invertible, i.e. $\det A=-a_{22}(d_{21}^2+c_{21}^2)/d_{12}\neq 0$. None of the corresponding transformations is canonical.
 
This type of Hamiltonians corresponds to a well known class of quadratically superintegrable systems  of the Minkowski plane, classified as Class II in \cite{daskaloyannis_unified_2006} (in this reference, manifolds and Hamiltonians are considered in general complex while we limit ourselves to the real case). 
%\subsection{}

We search now for canonoid transformations of the system (\ref{emb}) leading to Hamiltonians of the form $K_2$. It is possible in this case to consider the matrix $A$ in full generality. We find that the Hamiltonian of the transformed system, when the transformation is canonoid, is always in the form
\begin{equation}\label{K2}
	K_2=P_1P_2+k(\alpha_1Q_1-\alpha_2Q_2)^2,
\end{equation}
where $k$, $\alpha_1$ and $\alpha_2$ are  constants. After the point-transformation determined by
\begin{eqnarray*}
	X=\alpha_1Q_1-\alpha_2Q_2,\quad Y=\alpha_1Q_1+\alpha_2Q_2,
\end{eqnarray*}
we have
\[
	K_2=\alpha_1\alpha_2(-P_X^2+P_Y^2)+kX^2.
\]
We can divide the Hamiltonian by the constant $\alpha_1\alpha_2\ne 0$ and see that, similarly to the original one, admits two evident quadratic first integrals and the functionally independent local third integral
\[
	Y+P_Y\ln (X+P_X).
\]
In this case the local superintegrable structure remains unchanged. After the computation, we observe that  no canonoid transformation such that $\alpha_1\alpha_2=0$ do exist. Nevertheless, we can analyze the superintegrability of the system of Hamiltonian (\ref{K2}) in this case.
If, say, $\alpha_2=0$, then the system admits two evident quadratic first integrals plus the third-one
\[
	P_2(Q_2P_2-Q_1P_1)-\frac 23kQ_1^3,
\]
and the system is quadratically superintegrable.

As an example, for $c_{22}=c_{12}=0$, $c_{21}a_{22}=b_{21}d_{22}$, $d_{21}a_{22}=-a_{21}d_{22}$, $b_{12}d_{12}=d_{22}b_{22}$, $a_{12}d_{12}=d_{22}a_{22}$, $a_{11}=(-4d_{22}a_{22}b_{21}^2-d_{12}^2b_{11}b_{21}-4d_{22}a_{21}^2a_{22}+b_{21}^2d_{22}d_{12}+a_{21}^2d_{22}d_{12}+a_{22}d_{12}c_{11}b_{11}+a_{22}d_{11}d_{22}a_{21}-c_{11}d_{22}b_{21}a_{22})/(d_{12}(a_{22}d_{11}+a_{21}d_{12}))$, with $\det A=-4a_{22}d_{22}^2(a_{21}^2+b_{21}^2)/d_{12}\neq 0$, we have
\[
	K_2=P_1P_2+\frac{d_{12}}{2a_{22}d_{22}}(d_{12}Q_1-d_{22}Q_2)^2.
\]
A computation shows that none of the transformations leading to the last form of $K_2$ is canonical.
 
%We can summarize our findings by saying that the linear canonoid transformations of the systems considered above do preserve the existing quadratically integrable or superintegrable structure. The local superintegrable structure of (\ref{emb}) is also conserved. Moreover, we showed that linear canonoid transformations can map Hamiltonian systems of the Euclidean plane, superintegrable or not, into Hamiltonian systems of the Minkowski plane with the same superintegrable characteristics.

\section{Whittaker's characterization}

Since the first edition (1904) of  his celebrated Treatise on Analytical Mechanics \cite{whittaker_treatise_1988}, E. T. Whittaker characterizes what we call here canonoid transformations. Given a system of ODEs
\begin{equation}\label{ode}
	\frac d{dt}x^r=\mathcal{X} ^r(x^1,\ldots, x^n,t),\quad r=1,\ldots, n,
\end{equation}
and a one-form $M(x^r,t)$, the absolute and relative integral invariants of the differential equations are defined following Poincar\'e \cite{poincare_sur_1890}.  We do not need here to recall the definitions of integral invariants (for this, see \cite{whittaker_treatise_1988}, \S \S 112-116), but only their characterization in modern notation. We have that $M$ determines an absolute invariant integral if and only if
\[
	\frac \partial {\partial t}M+\mathcal{L} _{ \mathcal{X} } M=0,
\]
where $ \mathcal{L} _{ \mathcal{X} } M $ is the Lie derivative of $M$ along the vector field $ \mathcal{X} $. $M$ determines a relative invariant integral if and only if $\mathbf{d} M$ is an absolute integral invariant.
If the coordinates $(x^i)$ can be divided into two sets $(q^i,p_i)$, such that $n=2N$, then, as stated in \S 116 of \cite{whittaker_treatise_1988}, 
\begin{proposition}
	The ODEs (\ref{ode}) in coordinates $(q^i,p_i)$ are in Hamiltonian form if and only if
	\[
		\Sigma_{i=1}^Np_i\delta q^i,
	\]
	determines a relative invariant integral of (\ref{ode}).
\end{proposition}
Indeed, if we consider time-independent systems and if we identify the variational quantities $\delta q^i$ with the differentials $\mathbf{d} q^i$
then $\Sigma_{i=1}^Np_i\delta q^i$ becomes the Liouville one-form $\theta = \sum _{ i } p _i \mathbf{d} q ^i $, and $- \mathbf{d} \theta = \omega = \sum \mathbf{d} q ^i \wedge \mathbf{d} p _i $ becomes the symplectic form. Hence, by Cartan's magic formula, we have that
\begin{equation}\label{ri}
	\mathcal{L} _{ \mathcal{X} } \left( \mathbf{d}\theta \right) =-\mathcal{L} _{ \mathcal{X} } (\omega )=-\mathbf{d}  ({\bf i}_{\mathcal{X}}(\omega) ) - \mathbf{i} _{ \mathcal{X} } \mathbf{d} \omega =-\mathbf{d}  ({\bf i}_{\mathcal{X}}(\omega) )=0,
\end{equation}
and thus the vector field $ \mathcal{X} $ is locally Hamiltonian.
Moreover, if the manifold is contractible, thanks to the Poincar\'e lemma we have
\[
	{\bf i}_{\mathcal{X}}(\omega )=\mathbf{d} H,
\]
for some function $H$, that means that the system is Hamiltonian, and the previous statement follows in the case of the relative integral invariance condition.
For the absolute integral invariance we have
\begin{equation}\label{ai}
0 =     \mathcal{L} _{ \mathcal{X} } (\theta )={\bf i}_{\mathcal{X}} \mathbf{d} (\theta )+\mathbf{d} {\bf i}_{\mathcal{X}}(\theta )=-{\bf i}_{\mathcal{X}} \omega +\mathbf{d} {\bf i}_{\mathcal{X}}(\theta ),
\end{equation}
and the system is clearly Hamiltonian with Hamilton function ${\bf i}_{\mathcal{X}}(\theta )$.

Finally, in \S 136 of \cite{whittaker_treatise_1988}, the transformations of coordinates $(P_j(q^i,p_i), Q^j(q^i,p_i))$ that maintain the Hamiltonian form of (\ref{ode}), our canonoid transformations, are naturally characterized as those for which the form $\bf{PdQ}$ determines an invariant integral (relative or absolute) of the ODEs. Canonical transformations are defined  in the same section of \cite{whittaker_treatise_1988}.

This characterization provides a simple direct way to characterize the possible canonoid transformations, or, equivalently, the possible alternative Hamiltonian representations for the field $\mathcal X$. Given a system of Hamiltonian $H$ on a symplectic manifold with symplectic form $\omega$, such that ${\bf i}_{\mathcal{X}}\omega =\mathbf{d} H$, we can determine another local Hamiltonian structure for the field $\mathcal X$ whenever we know some non-closed one-form $ \Theta$, such that $\mathbf{d}  \Theta$ is  non degenerate, satisfying
\begin{equation}\label{dxm}
	\mathcal{L} _{ \mathcal{X} } \mathbf{d} \Theta=0.
\end{equation}
In this case, by (\ref{ri}), we know that, at least locally, $i_{\mathcal X}\mathbf {d} \Theta=\mathbf {d}K$ for some Hamiltonian function $K$. A stronger, global, condition is provided if $ \Theta$ is an absolute invariant integral with $\mathbf{d} \Theta$ non degenerate. In this case, by (\ref{ai})
 $ \Omega = -\mathbf{d}  \Theta$   is the new symplectic form and the new Hamiltonian $K$ of the system is
\[
	K={\bf i}_{\mathcal{X}} \Theta.
\]
In both cases, when we can write $ \Theta=P_i\mathbf {d}Q^i$ for some coordinate system $(P_i,Q^i)$, the transformation $(p_i,q^i)\leftrightarrow  (P_i,Q^i)$ is canonoid.
We remark that, if $ \Theta-\sum p _i \mathbf{d} q ^i  =\mathbf {d}f$ for some function $f$, then the transformation is the identity.

By putting $\mathcal{X} =\frac {\partial H}{\partial p_i} \frac \partial {\partial q^i}- \frac {\partial H}{\partial q^i}\frac \partial {\partial p_i}$, the condition $\mathcal{L}_{\mathcal{X}} \Theta=0$  becomes a system of first-order PDEs in $ \Theta_i(q^j,p_j)$ involving the Hamiltonian $H$.
These last two conditions are not very  different from those given for generating functions of canonoid transformations in Carinena and Ranada \cite{carinena_canonoid_1988}.

We can call the one-forms $ \Theta$ such that $\mathcal{L} _{ \mathcal{X} }  \Theta=0$ {\bf absolute generators} (or global generators) of a canonoid transformation. We call $ \Theta$ a {\bf relative generator} (or local generator) of a canonoid transformation when
\[
	\mathcal{L} _{ \mathcal{X} } \mathbf{d}  \Theta=0.
\]
If $(p_i,q^i)$ are canonical coordinates, then the Liouville one-form $ \Theta=\sum p _i \mathbf{d} q ^i $ is a relative generator of the identity transformation.

For example, if $H$ is the harmonic oscillator with coordinates $ (x ^1 , x ^2 , x ^3 , x ^4) = (q _1 , q _2 , p _1 , p _2) $
\[
	H=\frac 12(p_1^2+p_2^2+ q _1 ^2+ q _2   ^2),
\]
then
\[
	\mathcal{X} =p_1\frac \partial{\partial q _1  }+p_2\frac \partial{\partial q _2  }-q _1  \frac \partial{\partial p_1}-q _2  \frac \partial{\partial p_2}.
\]
Therefore, the absolute generators of canonoid transformations of the harmonic oscillator are characterized by
\[
	\mathcal{L} _{ \mathcal{X} }  \Theta=0 \Leftrightarrow \begin{cases} \mathcal{X} ( \Theta_1)= \Theta_3,\\ \mathcal{X} ( \Theta_2)= \Theta_4,\\ \mathcal{X} ( \Theta_3)=- \Theta_1,\\ \mathcal{X} ( \Theta_4)=- \Theta_2, \end{cases}
\]
with the evident integrability conditions $\mathcal{X} ^2  ( \Theta_i)=- \Theta_i$, $i=1,\ldots,4$.%, where $(x^1,x^2,x^3,x^4)=(q^1,q^2,p_1,p_2)$.

A solution is  $ \Theta_1=-p_1+q _2  $, $ \Theta_2=q _2 $, $ \Theta_3=q _1   +p_2$, $ \Theta_4=p_2$. We have in this case %$\mathcal{L} _{ \mathcal{X} }  \Theta=0$,
\[
	\mathbf{d}  \Theta = \begin{bmatrix}
	0 & -1 & 2 & 0  \\
	1 & 0 & 0 & 0 \\
	-2 & 0 & 0 & -1\\
	0 & 0 & 1 & 0
	\end{bmatrix},
\]
a non degenerate form with $\mathbf {d}(\mathbb {J}- \Theta)\neq 0$, and
\[
	K={\bf i}_{\mathcal{X}} \Theta=-(p_1^2+q _1 ^2+q _1 p_2-p_1q _2) ,
\]
that is a first integral of $H$. Then, the form $ \Theta$ and the function $K$ provide an alternative Hamiltonian structure for the harmonic oscillator. In the case when $H^1( M,\mathbb R)$ is zero, relative generators also determine global Hamiltonian structures.

\section{Poissonoid Transformations}
The following definition is a natural extension of the definition of canonoid transformations to the case of regular Poisson manifolds.

\begin{definition}
	Let  $ ( M , \pi ) $ be a  Poisson manifold,  and let $ \mathcal{X} $ be a locally Hamiltonian vector field on $M$, that is, for each $ x \in M $, there is a neighborhood $U$ of $x$ and  locally defined function $  H_U $ such that $ \mathcal{X} = \pi ^\sharp \mathbf{d}  H_U  $. A diffeomorphism %transformation
	$ f: M \to M $ is said to be {\bf Poissonoid} with respect to the vector field $ \mathcal{X} $ if the transformed vector field $ f _\ast \mathcal{X}$ is also locally Hamiltonian, that is, for each $ x \in M $, there is a neighborhood $V$ of $ f (x) $  and  a locally defined function $  K_V $ such that $ f _\ast  \mathcal{X} = \pi ^\sharp \cdot  \mathbf{d} K_V$.
\end{definition}
This is equivalent to saying that, for each $y \in M $, there is a neighborhood $V$ of $y$ and a locally defined function $ K _V $ such that $ \mathcal{X} = (f ^\ast \pi ^\sharp) \cdot \mathbf{d} (f ^\ast K _V) $.

\begin{remark}\label{rmk:poissonoid_locally_hamiltonian}
	By Proposition \ref{prop:Hamilton's equations I} the previous definition means that, for each point $ x \in M $ the system of equations associated with $ \mathcal{X} = \pi ^\sharp  \cdot \mathbf{d} H _U $ can be written, in splitting coordinates in the neighborhood $U$ of $x$, as:
	\begin{equation}\label{eqn:Hamilton_equations_H I}
		\dot q ^i  = \frac{ \partial H _U  } { \partial p _i } , \quad \dot p _i = -\frac{ \partial H _U  } { \partial q ^i  } , \quad \dot z ^k  = \sum _l \tilde \phi ^{ k l }(z) \frac{ \partial H _U  }{ \partial z ^l  }
	\end{equation}
	and the system associated with $ f _\ast \mathcal{X}  =  \pi ^\sharp\cdot  \mathbf{d} K _U $  can be written, choosing appropriate splitting coordinates in a neighborhood $V$ of $f(x)$, as:
	
	\begin{equation} \label{eqn:Hamilton_equations_K I}
		\dot Q ^i  = \frac{ \partial K _U  } { \partial P _i } , \quad \dot P _i =- \frac{ \partial K _U  } { \partial Q ^i  }, \quad \dot Z ^k  = \sum _l \tilde \phi ^{ k l }(Z) \frac{ \partial K _U  }{ \partial Z ^l  }
	\end{equation}
	so that the transformation $f$ carries the system of Hamilton's equations \ref{eqn:Hamilton_equations_H I} again into a system of Hamilton's equations \ref{eqn:Hamilton_equations_K I}. 
\end{remark}
 Note that, in the case of a Poisson manifold, not all  Poisson vector fields are locally Hamiltonian. This means that if $f $ is a Poissonoid map then $ \mathcal{L} _{ \mathcal{X} } (f ^\ast \pi) = 0 $, but the converse is in general not true. For instance, if $ \pi $ is the trivial Poisson bivector, then any  diffeomorphism pushes $ \pi $ to the trivial bivector and $ \mathcal{L} _{ \mathcal{X} } (f ^\ast \pi) = 0 $, for  any $ \mathcal{X}$. On the other hand, the only locally Hamiltonian vector field is the trivial one. If $ \mathcal{L} _{ \mathcal{X} } (f ^\ast \pi) = 0 $, we say that the map is {\bf weakly Poissonoid} with respect to the vector field $\mathcal{X} $. Weakly Poissonoid maps, in general, do not lead to the nice structure described in Remark \ref{rmk:poissonoid_locally_hamiltonian}.  
 \begin{remark} Since Poissonoid transformations are diffeomorphisms, by Corollary \ref{cor:compatible_poisson} they send compatible Poisson bivectors into compatible Poisson bivectors. This may be of use in finding hierarchies of (compatible) Poisson structures.  
 \end{remark}
 \subsection{Poissonoid transformations and integrals of motion.}
The definition of Poissonoid transformations can be specialized to Hamiltonian systems instead of  locally-Hamiltonian ones: If $ \mathcal{X}  $ is a Hamiltonian vector field on a Poisson manifold $ (M , \pi) $, a diffeomorphism  $ f : M \to M $ is a Poissonoid transformation with respect to $ \mathcal{X} $ if the transformed field $ f _\ast \mathcal{X} $ is also Hamiltonian with respect to $ \pi $, that is if there is a smooth function $K $ on $M$ such that $ f _\ast \mathcal{X} = \pi ^\sharp \cdot  \mathbf{d} K $.

Since $ f $ is a diffeomorphism, we have that  the vector field $ f _\ast \mathcal{X} $ is Hamiltonian with respect to $ \pi $ if and only if $ \mathcal{X} $ is Hamiltonian with respect the transformed bivector $ f ^\ast \pi $, i.e. there exists a smooth function $ K' $ on $M$ such that 
\[
    \mathcal{X} = (f ^\ast \pi ^\sharp )\cdot  \mathbf{d} ( K').   
\] 
This means that, if $ f $ is a Poissonoid transformation for $ \mathcal{X}$, then $ \mathcal{X} $ admits a new and possibly different Hamiltonian structure. If, in addition $ \pi $ and $ f ^\ast \pi $ are compatible, then the vector field $ \mathcal{X} $ will be  biHamiltonian.
%, that is, it is Hamiltonian with respect to two different Poisson bivectors: $ \pi $ and $ f ^\ast \pi $.
Under some additional conditions these facts are enough to show that the system has a complete set of integrals in involution.

More explicitly, the existence of the additional Poisson bivector $ f ^\ast \pi $, provides a concrete way of obtaining additional constants of motion from the Casimirs of  $ f ^\ast \pi $
\begin{proposition} Let $ (M , \pi) $ a Poisson manifold, and let $ \{~ ,~\} $ the Poisson bracket associated with $\pi$. Suppose $ \mathcal{X} $ is an Hamiltonian vector field, so that there exist a function $H$ such that $ \mathcal{X} = \pi ^\sharp \cdot \mathbf{d} H  $, and suppose that $f$ is a Poissonoid transfomation for $ \mathcal{X} $, so that there exist a function $ K' $ such that $ \mathcal{X} = (f ^\ast \pi ^\sharp )\cdot  \mathbf{d} ( K')$. Then $ K ' $ and any Casimir of the Poisson bivector $\pi $ are constants of motion of the Hamiltonian system $ (M , \pi , H)$. 
\end{proposition} 

\begin{proof}
   Since $ f $ is a Poissonoid transformation for $ \mathcal{X} $ we have that 
   \[
       \mathcal{X} = \pi ^\sharp \cdot \mathbf{d} H = (f ^\ast \pi ^\sharp) \cdot \mathbf{d} K '.
   \] 
   Suppose $ C $ is a Casimir of $ (f ^\ast \pi)$ , and that $ \{ ~,~ \} ' $ is the Poisson bracket associated to  $ (f ^\ast \pi)$, then
   \begin{align*} 
      0 =& \{ C, K ' \}' =  (f ^\ast \pi) (\mathbf{d} C, \mathbf{d} K' ) \\
        =&  \left\langle  \mathbf{d} C, (f ^\ast \pi ^\sharp)\cdot  \mathbf{d} K' \right \rangle   =  \left\langle  \mathbf{d} C, \mathcal{X}  \right \rangle\\
        =&   \left\langle  \mathbf{d} C,  \pi ^\sharp \cdot \mathbf{d} H \right\rangle=  \pi (\mathbf{d} C, \pi ^\sharp \cdot \mathbf{d} H  \rangle\\
        =&  \{ C, H \}  
  \end{align*} 
  Hence, $\{ C , H \} = 0 $, and $C$ is a constant of motion of the Hamiltonian system $ (M , \pi , H) $. The proof that $ K ' $ is a costant of motion is similar.
  %Similarly,
  % \begin{align*} 
  %    0 =& \{ K', K ' \}' =  (f ^\ast \pi) (\mathbf{d} K', \mathbf{d} K' ) \\
  %      =&  \left\langle  \mathbf{d} K', (f ^\ast \pi ^\sharp)\cdot  \mathbf{d} K' \right \rangle   =  \left\langle  \mathbf{d} K', \mathcal{X}  \right \rangle\\
  %      =&   \left\langle  \mathbf{d} K',  \pi ^\sharp \cdot \mathbf{d} H \right\rangle=  \pi (\mathbf{d} K', \pi ^\sharp \cdot \mathbf{d} H  \rangle\\
  %      =&  \{ K', H \}  
  %\end{align*} 
\end{proof} 
 
 Another important property is that Poissonoid transformations preserve constants of motion:
 
     \begin{proposition}\label{prop:Poissonoid_tr_preserve_constant_of_motion}
    Let $ (M, \pi , H )$ be a Hamiltonian system, and let $ f $ be a Poissonoid transformation such that $  f _\ast \mathcal{X} _{ H } = \pi ^\sharp \cdot \mathbf{d} K $, then  $ F $ is a constant of motion of the transformed system if and only if  $ f ^\ast F $ is a constant of motion for $ (M , \pi , H) $.          
     \end{proposition} 
 
     \begin{proof} \[\frac{ dF } { dt } = \{ F, K \} =  \left\langle \mathbf{d} F , \pi ^\sharp \cdot \mathbf{d} K \right\rangle = \left\langle \mathbf{d} F , f _\ast \mathcal{X} _H \right\rangle = \left\langle \mathbf{d} F , f _\ast \pi ^\sharp \cdot \mathbf{d} H \right\rangle \]
             since $ \mathcal{X} _K = f _\ast \mathcal{X} _H $.  Taking the pull-back of the above yields 
             \[f ^\ast\left(   \frac{ d   F } { dt }\right)  =  \frac{ d  f ^\ast F } { dt } = f ^\ast \left\langle \mathbf{d} F, f _\ast ( \pi ^\sharp  \cdot \mathbf{d} H) \right\rangle = \left\langle \mathbf{d} (f ^\ast F) , \pi ^\sharp \cdot \mathbf{d} H \right\rangle = \{ f ^\ast F, H \}.
             \]
             The proof follows.
     \end{proof} 
     \begin{remark}
        Since the transformation $f$ is a diffeomorphism, functionally independent constants of motion are sent into functionally independent constants of motion.  In the case when the Poisson structure coincides with a symplectic structure, Proposition \ref{prop:Poissonoid_tr_preserve_constant_of_motion} applies to canonoid transformations. 
     \end{remark} 
\begin{example}[Euler's Equations for the rigid body]
 %\begin{example}[Euler's Equations for the rigid body]
Let  $ \mathfrak{ so}(3) $ be the Lie algebra of $ SO (3) $, the group of rotations in $ \mathbb{R}^3  $, and let $ \pi $ be the Poisson tensor associated to the Lie-Poisson bracket.  Then $ (M, \pi) $, is a Poisson manifold. On the manifold $ M = \mathfrak{ so}^\ast (3) $ we introduce the coordinates $ \mathbf{m} = (m _1 , m _2 , m _3)   \in \mathfrak{ so}^\ast (3) \cong \mathbb{R}  ^3 $.
In these notations the Poisson tensor has the form
\[\pi _{ \mathbf{m} }  = \begin{bmatrix}
	0 & -m _3  & m _2 \\
	m _3  & 0&-m _1\\
	-m _2 &m _1 &0
	\end{bmatrix}.
\]
Euler's rigid body equations, are the Hamiltonian equations on $ (\mathfrak{ so}^\ast (3) , \pi) $  with Hamiltonian function
\[
	H = \frac{1}{2} \left( \frac{ m _1^2 } { I _1 } + \frac{ m _2 ^2 } { I _2 } + \frac{ m _2 ^2 } { I _3 } \right),
\]
where $ I _1 , I _2 , I _3 $ are the principal moments of inertia of the rigid body.
The corresponding Hamiltonian vector field, given by $ \mathcal{X} = \pi ^\sharp \cdot \mathbf{d} H$  is:
\[\mathcal{X}=
	\begin{bmatrix}[2.2]
		\dfrac{ I _2 - I _3 } { I _2 I _3 } m _2 m _3 \\
		\dfrac{ I _3 - I _1 } { I _3 I _1 } m _3 m _1 \\
		\dfrac{ I _1 - I _2 } { I _1 I _2 } m _1 m _2
	\end{bmatrix}
\]
Let $ N = \mathfrak{ so}^\ast (3) $ we introduce the coordinates $ \mathbf{n} = (n _1 , n _2 , n _3)   \in \mathfrak{ so}^\ast (3) \cong \mathbb{R}  ^3 $, then $ (N , \pi _{ \mathbf{n} }) $ is a Poisson manifold.
Consider the diffeomorphism  $ f : M \to N $ such that $ (n _1 , n _2 , n _3 ) = f (m _1 , m _2 , m _3) $ defined by the equations $n _1 = a m _1 , n _2 = b m _2 , n _3 = c m _3$, where $a$, $b$ and $c$ are non-zero constants.
The pushforward of $\mathcal{X}$ can be expressed as
\[f _\ast \mathcal{X}=
	\begin{bmatrix}[2.2]
		\dfrac{ (I _2 - I _3)a } { I _2 I _3bc } n _2 n _3 \\[1em]
		\dfrac{ (I _3 - I _1)b} { I _3 I _1ac } n _3 n _1  \\
		\dfrac{ (I _1 - I _2)c } { I _1 I _2ab } n _1 n _2
	\end{bmatrix}
\]
If $ a = \sqrt { I _2 I _3 }, b = \sqrt { I _1 I _3 } $   and $ c = \sqrt { I _1 I _2 } $

\[f _\ast \mathcal{X}=
	\begin{bmatrix}[2.2]
		\dfrac{ (I _2 - I _3) } {I _1  I _2 I _3 } n _2 n _3 \\
		\dfrac{ (I _3 - I _1)} { I _1 I _2 I _3  } n _3 n _1 \\
		\dfrac{ (I _1 - I _2) } { I _1 I _2 I _3 } n _1 n _2
	\end{bmatrix}
\]
and one possible corresponding Hamiltonian, obtained from $f _\ast \mathcal{X} = \pi _{ \mathbf{n} } \mathbf{d} K $, is
\[K =- \dfrac{1}{2} \left( \dfrac{ n _1 ^2 } { I _2 I _3 } + \dfrac{ n _2 ^2 } { I _1 I _3 } + \dfrac{ n _3 ^2 } { I _1 I _2 } \right).\]

Pulling back K we get $ f ^\ast K = -\frac{1}{2} ( m _1 ^2 + m _2 ^2 + m _2 ^2) $. Pulling back the form $ \pi _{ \mathbf{n} } $ yields
\[ f ^\ast \pi _{ \mathbf{n} } =
	\begin{bmatrix}[2]
		0                        & - \dfrac{ m _3 } { I _3 } & \dfrac{ m _2}{I _2 }     \\
		\dfrac{ m _3 } { I _3 }  & 0                         & -\dfrac{ m _1 } { I _1 } \\
		-\dfrac{ m _2 } { I _2 } & \dfrac{ m _1 } { I _1 }   & 0
	\end{bmatrix}
\]
and thus $ \mathcal{X} = \pi _{ \mathbf{m} }  ^\sharp \cdot \mathbf{d} H = (f ^\ast \pi _{ \mathbf{n} } ) ^\sharp  \cdot \mathbf{d} (f ^\ast  K)$. This shows that the rigid body equations are biHamiltonian.
\end{example} 
 \begin{example}[Euler's equations on $ \mathfrak{ so} ^\ast  (4) $]  
 %\begin{example}[Euler's equations on $ \mathfrak{ so} ^\ast  (4) $]
Here we use the same notations as in \cite{falqui_note_2007}. The manifold  $ M = \mathfrak{ so} ^\ast (4)$ is six dimensional.
Since $ \mathfrak{ so }(4) $ is isomorphic to the space of $ 4 \times 4 $ skew-symmetric matrices, identifying $ \mathfrak{ so}^\ast (4) $ with $ \mathfrak{ so}(4) $  we can write any element of $ \mathfrak{ so}^\ast (4) $ as
\[
	\mathcal{M} =  \sum_{ i <j = 1 } ^4  m _{ ij } \left( E _{ ij } - E _{ ji }\right).
\]
where  $ E _{ ij } $ denotes the elementary matrix whose $ (i,j) $ entry is 1. The manifold
$ M = \mathfrak{ so}^\ast (4) $ is endowed with a Lie-Poisson structure that, in the  variables $\mathbf{m} = (m _{ 12 }, m _{ 13 } , m _{ 14 } , m _{ 23 } , m _{ 24 } , m _{ 34 })$, can be written as
\[
	\pi _{ \mathbf{m} }  = \begin{bmatrix}
	0 & - m _{ 23 }  & - m _{ 24 } & m _{ 13 } & m _{ 14 } & 0 \\
	m _{ 23 } & 0 & - m _{ 34 } & - m _{ 12 } & 0 & m _{ 14 } \\
	m _{ 24 } & m _{ 34 } & 0 & 0 & - m _{ 12 } & - m _{ 13 } \\
	- m _{ 13 } & m _{ 12 } & 0 & 0 & - m _{ 34 } & m _{ 24 } \\
	- m _{ 14 } & 0 & m _{ 12 } & m _{ 34 } & 0 & - m _{ 23 } \\
	0 & - m _{ 14 } & m _{ 13 } & - m _{ 24 } & m _{ 23 } & 0
	\end{bmatrix}.
\]
The Euler-Manakov equations for the rigid body on $ (\mathfrak{ so}^\ast (4) , \pi _{ \mathbf{m} }) $ are the Hamiltonian equations with the following quadratic Hamiltonian
\[
	H = \frac{1}{2} \sum _{ i < j } a _{ ij } m _{ ij} ^2
\]
where the coefficients $ a _{ ij } $ can be written as
\[
	a _{ ij } = J _l ^2 + J _k ^2
\]
with $ \{ i,j,l,k \} $ a permutation of $ \{ 1,2,3,4 \} $.
The rank of the $ \mathfrak{ so } ^\ast  (4) $ Lie-Poisson structure is $ 4 $ almost everywhere, the Casimirs are
\[
	C _1 = m _{ 12 } ^2 + m _{ 13 } ^2 + m _{ 14 } ^2 + m _{ 23 } ^2 + m _{ 24 } ^2 + m _{ 34 } ^2,  \quad C _2 = m _{ 12 } m _{ 34 } + m _{ 14 } m _{ 23 } - m _{ 13 } m _{ 24 }
\]
where $ C _1 = - \frac{1}{2} \operatorname{Tr} (\mathcal{M}) $, and $ C _2 = \operatorname{Pf }(\mathcal{M})$ (with $ \operatorname{Pf}  (\mathcal{M})$  the Pfaffian of $ \mathcal{M} $).
The corresponding Hamiltonian vector field, given by $ \mathcal{X} = \pi ^\sharp  \cdot \mathbf{d} H $ is
\[
	\mathcal{X} = \begin{bmatrix}[1.3]
	J _1 ^2 (m _{ 13 } m _{ 23 } + m _{ 14 } m _{ 24 }) - J _2 ^2 (m _{ 13 } m _{ 23 } + m _{ 14 } m _{ 24 })     \\
	J _3 ^2 (m _{ 12 } m _{ 23 } - m _{ 14 } m _{ 34 }) + J _{ 1 } ^2 (m _{ 14 }  m _{ 34 } - m _{ 12 } m _{ 23 })    \\
	J _4 ^2 ( m _{ 12 } m _{ 24 } + m _{ 13 } m _{ 34 }) - J _1 ^2  (m _{ 12 } m _{ 24 } + m _{ 13 } m _{ 34 }) \\
	J _2 ^2 (m _{ 12 } m _{ 13 } + m _{ 24 } m _{ 34 }) - J _3 ^2 (m _{ 12 } m _{ 13 } + m _{ 24 } m _{ 34 }) \\
	J _4 ^2 (m _{ 23 } m _{ 34 } - m _{ 12 } m _{ 14 }) + J _2 ^2 (m _{ 12 } m _{ 14 } - m _{ 23 } m _{ 34 }) \\
	J _3 ^2 (m _{ 13 } m _{ 14 } + m _{ 23 } m _{ 24 }) - J _4 ^2 (m _{ 13 } m _{ 14 } + m _{ 23 } m _{ 24 })
	\end{bmatrix}
\]
Let $ N = \mathfrak{ so} (4) $ we introduce the coordinates $ \mathbf{n} = (n _{ 12 } , n _{ 13 } , n _{ 14 } ,n _{ 23 } , n _{ 24 } , n _{ 34 }  ) $ then $ (N , \pi _{ \mathbf{n} }) $ is a Poisson manifold.
Consider the diffeomorphism  $ f : M \to N $ defined by the equations
\begin{align*}
	n _{ 12 } & =  \frac{ m _{ 12 }} { J _1 J _2 }   ,~ n _{ 13 }  =  \frac{ m _{ 13 } } { J _1 J _3 }   , ~ n _{ 14 }  = \frac{  m _{ 14 }} { J _1 J _4 } \\
	n _{ 23 } & =   \frac{ m _{ 23 }} { J _2 J _3 }  , ~n _{ 24 }  = \frac{ m _{ 24 }} { J _2 J _4 }    , ~ n _{ 34 }  = \frac{ m _{ 34 }} { J _3 J _4 }
\end{align*}
The pushforward of $\mathcal{X}$ can be expressed as
\[
	f _\ast  \mathcal{X} = \begin{bmatrix}[1.3]
	J _1 ^2 (J _3 ^2 n _{ 13 } n _{ 23 } +J _4 ^2  n _{ 14 } n _{ 24 }) - J _2 ^2 (J _3 ^2 n _{ 13 } n _{ 23 } + J _4 ^2  n _{ 14 } n _{ 24 })     \\
	J _3 ^2 (J _2 ^2 n _{ 12 } n _{ 23 } - J _4 ^2 n _{ 14 } n _{ 34 }) + J _{ 1 } ^2 (J _4 ^2 n _{ 14 }  n _{ 34 } -J _2 ^2  n _{ 12 } n _{ 23 })    \\
	J _4 ^2 ( J _2 ^2 n _{ 12 } n _{ 24 } + J _3 ^2 n _{ 13 } n _{ 34 }) - J _1 ^2  (J _2 ^2 n _{ 12 } n _{ 24 } + J _3 ^2 n _{ 13 } n _{ 34 }) \\
	J _2 ^2 (J _1 ^2 n _{ 12 } n _{ 13 } + J _4 ^2 n _{ 24 } n _{ 34 }) - J _3 ^2 (J _1 ^2 n _{ 12 } n _{ 13 } + J _4 ^2 n _{ 24 } n _{ 34 }) \\
	J _4 ^2 (J _3 ^2 n _{ 23 } n _{ 34 } - J _1 ^2  n _{ 12 } n _{ 14 }) + J _2 ^2 (J _1 ^2 n _{ 12 } n _{ 14 } -J _3 ^2  n _{ 23 } n _{ 34 }) \\
	J _3 ^2 (J _1 ^2 n _{ 13 } n _{ 14 } + J _2 ^2 n _{ 23 } n _{ 24 }) - J _4 ^2 (J _1 ^2 n _{ 13 } n _{ 14 } + J _2 ^2 n _{ 23 } n _{ 24 })
	\end{bmatrix}
\]
and one possible Hamiltonian, corresponding to $f _\ast \mathcal{X} $, and obtained from $f _\ast \mathcal{X} = \pi _{ \mathbf{n} } \mathbf{d} K $, is
\[K =- \dfrac{1}{2} \left(J _1 ^2 J _2 ^2 n _{ 12 } ^2 + J _1^2  J _3^2  n _{ 13 } ^2 + J _1 ^2  J _4 ^2 n _{ 14 } ^2 + J _2 ^2 J _3 ^2 n _{ 23 } ^2 + J _2 ^2 J _4 ^2 n _{ 24 } ^2 + J _3 ^2 J _4 ^2 n _{ 34 } ^2  \right).\]
Pulling back K we get
\[f ^\ast K = -\frac{1}{2} ( m _{ 12 } ^2 + m _{ 13 } ^2 + m _{14 } ^2 + m _{ 23 } ^2 + m _{ 24 } ^2 + m _{ 34 } ^2), \]
a Casimir of $ \pi_{ \mathbf{m} }  $. Pulling back the form $ \pi _{ \mathbf{n} } $ yields
\[
	f ^\ast \pi _{ \mathbf{n} }  = \begin{bmatrix}[1.3]
	0 & - J _1 ^2 m _{ 23 }  & - J _1 ^2 m _{ 24 } & J _2 ^2 m _{ 13 } & J _2 ^2 m _{ 14 } & 0 \\
	J _1 ^2 m _{ 23 } & 0 & - J _1 ^2 m _{ 34 } & - J _3  ^2 m _{ 12 } & 0 & J _3  ^2 m _{ 14 } \\
	J _1 ^2 m _{ 24 } & J _1 ^2 m _{ 34 } & 0 & 0 & - J _4 ^2 m _{ 12 } & - J _4 ^2 m _{ 13 } \\
	- J _2 ^2 m _{ 13 } & J _3 ^2 m _{ 12 } & 0 & 0 & - J _2 ^2 m _{ 34 } & J _3 ^2 m _{ 24 } \\
	- J _2 ^2 m _{ 14 } & 0 &J _4 ^2  m _{ 12 } & J _2 ^2 m _{ 34 } & 0 & - J _4 ^2 m _{ 23 } \\
	0 & - J _3 ^2 m _{ 14 } & J _4 ^2 m _{ 13 } & - J _3 ^2 m _{ 24 } & J _4 ^2 m _{ 23 } & 0
	\end{bmatrix}.
\]
Therefore, $ \mathcal{X} = \pi _{ \mathbf{m} }  ^\sharp\cdot  \mathbf{d} H = (f ^\ast \pi _{ \mathbf{n} } ) ^\sharp  \cdot \mathbf{d} (f ^\ast  K)$. This computation recovers the Poisson structure obtained in \cite{bolsinov_compatible_1992,morosi_euler_1996,falqui_note_2007}. It can be shown that this Poisson structure is compatible with the original one. Thus, we rediscovered that the Euler-Manakov equations of motion admit a bihamiltonian formulation (compare our expressions with \cite{falqui_note_2007}). Using this additional Poisson structure it is possible to prove the integrability of the Euler-Manakov equations (see \cite{bolsinov_compatible_1992,morosi_euler_1996,falqui_note_2007}), in fact they admit one additional integral of motion
\[
	I _1 = J _3 ^2 J _4 ^2 m _{ 12 } ^2 + J _2 ^2 J _4 ^2 m _{ 13 } ^2+ J _2 ^2 J _3 ^2 m _{ 14 } ^2+ J _1 ^2 J _4 ^2 m _{ 23 } ^2+ J _1 ^2 J _3 ^2 m _{ 24 } ^2+ J _1 ^2 J _2 ^2 m _{ 34 } ^2.
\]
\end{example} 
\begin{example}[Kirchhoff's equations]
%\subsection{Kirchhoff's equations}
In the previous examples it was possible to find a biHamiltonian structure by finding a Poissonoid transformation. In those examples, the transformation was a rescaling, and thus it was quite simple. We now we give an example where the Poissonoid transformation is more complex.

The motion of a rigid body in an ideal fluid can be described by the so called {\bf Kirchhoff's equations}.  These equations can be written as Hamiltonian equations on $ \mathfrak{ e } ^\ast (3)$  (the Lie algebra of the group $ E (3) $ of motions of the three-dimensional Euclidean space) with  the Lie-Poisson bracket. The Lie algebra $ \mathfrak{ e}(3) $ is a semidirect sum of $ \mathfrak{ so}(3) $ and  the group of translations in $\mathbb{R}^3$, that is $ \mathfrak{ e }(3) = \mathfrak{ so }(3) \oplus \mathbb{R}  ^3 $. On the manifold $ \mathcal{M} \cong \mathfrak{ e}^\ast (3) $, we introduce the coordinates $ \mathbf{z} = ( \mathbf{p} , \mathbf{m} ) $, where
\[
	\mathbf{p} = ( p _1 , p _2 , p _3 ) \in \mathbb{R}^3, \quad \mathbf{m} = ( m _1 , m _2 , m _3 ) \in \mathbb{R}  ^3 \cong \mathfrak{ so}(3)
\]
are two three dimensional vectors. The hat map $ \hat~ :\mathbb{R}^3   \to \mathfrak{ so}(3) $ defined as
\[
	\mathbf{v} = (v _1 , v _2 , v _3 ) \to  \mathbf{\hat v}  =  \begin{bmatrix}
	0 & -v _3  & v _2 \\
	v _3  & 0&-v _1\\
	-v _2 &v _1 &0
	\end{bmatrix}
\]
defines an isomorphism of the Lie algebras $(\mathbb{R}^3  , \times) $ and $ (\mathfrak{ so}(3) , [,]) $. If we identify the Lie algebra $ \mathfrak{ e}(3) $ with its dual, in these notations, the Lie-Poisson tensor has the form
\[
	\pi  = \begin{bmatrix}
	0 & 0  & 0& 0 & - p _3  & p _2  \\
	0 & 0 & 0 &  p _3  & 0 & - p _1  \\
	0 & 0 & 0 & - p _2  &  p _1  & 0 \\
	0 & - p _3  &  p _2  & 0 & - m _3 & m _2 \\
	p _3  & 0 &- p _1  & m _3 & 0 & - m _1 \\
	- p _2  & p _1  & 0 & - m _2  & m _1  & 0.
	\end{bmatrix}
\]
This Poisson tensor has the following quadratic Casimirs:
\[
	C _1 = p _1 ^2 + p _2 ^2 + p _3 ^2, \quad C _2 = m _1 p _1 + m _2 p _2 + m _3 p _3.
\]
Hamilton's equations corresponding to a quadratic Hamiltonian are called Kirchhoff equations.

A famous integrable case of the Kirchhoff equations was discovered by  Clebsch and it is characterized by the Hamiltonian
\[
	H _1 = \frac{1}{2} \left(  m _1 ^2 + m _2 ^2 + m _3 ^2 + \omega _1 p _1 ^2 + \omega _2 p _2 ^2 + \omega _3 p _3 ^2 \right)
\]
The Hamiltonian vector field in this case is
\[
	\mathcal{X} = \pi ^\sharp \cdot \mathbf{d} H _1  =
	\begin{bmatrix}[1.3]
		m _3 p _2 - m _2 p _3               \\
		m _1 p _3 - m _3 p _1               \\
		m _2 p _1 - m _1 p _2               \\
		( \omega _3 - \omega _2 ) p _2 p _3 \\
		( \omega _1 - \omega _3 ) p _1 p _3 \\
		( \omega _2 - \omega _1 ) p _1 p _2
	\end{bmatrix}
\]
and an additional integral of motion is
\[
	I = \frac{1}{2} \left( \omega _1 m _1 ^2 + \omega _2 m _2 ^2 + \omega _3 m _3 ^2 - \omega _2 \omega _3 p _1 ^2 - \omega _3 \omega _1 p _2 ^2 - \omega _1 \omega _2 p _3 ^2 \right).
\]

A Poisson bivector compatible with $\pi$  is
\[
	\widetilde \eta  =
	\begin{bmatrix}
		0      & - m _3                         & m _2                           & 0                               & 0                              & 0                              \\
		m _3   & 0                              & - m _1                         & (\omega _1 - \omega _2  ) p _3  & 0                              & (\omega _2 - \omega _1)   p _1 \\
		- m _2 & m _1                           & 0                              & (\omega _3 - \omega _1)    p _2 & (\omega _1 - \omega _3)   p _1 & 0                              \\
		0      & (\omega _2 - \omega _1)   p _3 & (\omega _1 - \omega _3)   p _2 & 0                               & (\omega _3 - \omega _1) m _3   & (\omega _1 - \omega _2) m _2   \\
		0      & 0                              & (\omega _3 - \omega _1)   p _1 & (\omega _1 - \omega _3) m _3    & 0                              & 0                              \\
		0      & (\omega _1 - \omega _2)   p _1 & 0                              & (\omega _2 - \omega _1) m _2    & 0                              & 0
	\end{bmatrix}
\]
Note that a linear change of variables transforms the Euler equations on $ \mathfrak{ so}^\ast (4) $  to equations $\mathfrak{ e}^\ast (3) $, which in the case of a positive definite quadratic Hamiltonian are the Kirchhoff's equations describing the motion of a rigid body in an ideal fluid \cite{bolsinov_compatible_1992}. Thus, in principle, the  Poisson bivector above  and the Poissonoid transformation we find below could be obtained from the previous example. However, we prefer to find them with a direct computation.

In this case, it is easy to verify  that $ \mathcal{X} = \pi ^\sharp\cdot  \mathbf{d} H _1 = (\widetilde \eta ) ^\sharp \cdot \mathbf{d} ( - \frac{1}{2} C _1  ) $. It can also be shown that this Poisson structure is compatible with the original one. Hence, the Clebsch system is biHamiltonian.

%\begin{framed}
	In analogy with the previous examples, we search now for a transformation $\phi  = f ^{ - 1 }  $ such that
	\begin{equation}\label{pb}
		-\frac 12\tilde \eta=f ^*\pi_F,
	\end{equation}
	(we consider here $-\frac 12 \tilde \eta$ instead of $\tilde \eta$ for computational convenience) and, consequently, $\mathcal{X}=(f ^*\pi_F) d(C_1)$, where $F^1 ,\ldots, F^6 $ are new coordinates and
	\[
		\pi_F  = \begin{bmatrix}
		0 & 0  & 0& 0 & - F ^3   & F ^2   \\
		0 & 0 & 0 &  F ^3   & 0 & - F ^1   \\
		0 & 0 & 0 & - F ^2   &  F ^1   & 0 \\
		0 & - F ^3   &  F ^2   & 0 & - F ^6  & F ^5   \\
		F ^3   & 0 &- F ^1   & F ^6  & 0 & - F ^4  \\
		- F ^2   & F ^1   & 0 & - F ^5   & F ^4   & 0
		\end{bmatrix}.
	\]
	While in all the previous examples $f$ was a simple rescaling of the old coordinates, here we must assume  that $f$, and thus $ \phi $, are linear and homogeneous in the new coordinates $F ^i$, so that
	\begin{equation} \label{eqn:definition_of_f}
        \phi : \phi ^i=p_i=a^i_jF^j, \qquad \phi ^{i+3}=m_i=a^{i+3}_jF^j, \quad i=1,\ldots ,3; \, j=1,\ldots, 6;
	\end{equation}
	with $a^k_j$ constants. The previous relations allow to write the components of $\tilde \eta$ with respect to the $(\phi ^j)$ as functions of the $(F^j)$, therefore, we can solve the equation (\ref{pb}) without inverting the transformation $\phi $. Substituting   $ \phi =f ^{ -1 } $ and  $ X=F  $ in equation \eqref{eqn:pull_back_bivector},  and using \eqref{eqn:definition_of_f}, yields:
	\begin{align*}
		( f ^\ast \pi) ^{ jk } & = \left( \frac{ \partial \phi  ^j } { \partial F ^r } \circ f \right)  \left( \frac{ \partial \phi  ^k } { \partial F ^s } \circ f \right) \pi ^{ rs } \circ f \\
		                          & =  \left( a ^j _r \circ f \right)  \left( a ^k _s  \circ f \right) \pi ^{ rs } \circ f                                                                \\
		                          & =   a ^j _r  a ^k _s   \pi ^{ rs } \circ (\phi  ^{ - 1 }  )
	\end{align*}
    Therefore, condition (\ref{pb})  written explicitly  in local coordinates is
    \[  -\frac 12 \tilde \eta ^{ jk }(x) =  a ^j _r  a ^k _s   \pi ^{ rs }( f (x)) =   a ^j _r  a ^k _s   \pi ^{ rs }  (\phi ^{ - 1 }(x)  )\]
	or, if we take $ x = \phi (F) $
	\[
		-\frac 12 \tilde \eta ^{ jk }(\phi (F)) =     a ^j _r  a ^k _s   \pi ^{ rs }  (\phi  ^{ - 1 }(\phi  (F))  )=  a ^j _r  a ^k _s   \pi ^{ rs }  (F )
	\]
that can also be written as
\[
	-\frac 12 \tilde \eta ^{jk}=a^j_ra^k_s\pi_F^{rs}
\]
	This last equation can be solved for the $a^j_k$ giving for example the following linear transformation
	\[
		\begin{bmatrix} p_1 \\
			p_2 \\
			p_3 \\
			m_1 \\
			m_2 \\
			m_3
		\end{bmatrix}=\begin{bmatrix}
		a & 0 & -\frac {a}{2\epsilon}B & 0 & 0 & 0 \\
		0 & 0 &  0 & 0 & \frac 12 A & 0 \\
		0 & 0 & 0 & \frac 12 B & 0 & \epsilon \\
		0 & 0 & 0 & -\epsilon A & 0 & \frac 12 AB\\
		0 & -\frac{a}{2\epsilon}C & 0 & 0 & 0 & 0 \\
		-\frac{a}{2\epsilon}AB & 0 & -aA & 0 & 0 & 0
		\end{bmatrix}
		\begin{bmatrix}
			F_1 \\
			F_2 \\
			F_3 \\
			F_4 \\
			F_5 \\
			F_6
		\end{bmatrix},
	\]
	where $a$, $\epsilon$ are real parameters, $A=\sqrt{\omega_1-\omega_2}$, $B=\sqrt{\omega_1-\omega_3-4\epsilon^2}$, $ C = \omega _1 - \omega _3 $, and whose determinant is
	\[
	-\left(\frac{a}{4\epsilon}AC \right)^3.
	\]
	If $\omega_1>\omega_2>\omega_3\geq 0$ then a positive  value of  $\epsilon$ can always be found such that the transformation is real. In the new coordinates we have
    \begin{align*}
	C_1= & a^2(F^1)^2+\left(\frac a{2\epsilon}B\right)^2(F ^3 )^2+\frac 14 B^2(F^4 )^2+\frac 14A^2(F ^5 )^2\\
    & +\epsilon^2(F ^6 )^2-\frac {a^2}\epsilon B F ^1 F ^3 +\epsilon B F ^4 F ^6 .
\end{align*}
\end{example} 
%%%%%%%%%%%%%%%%%%%%%%%
%
%%%%%%%%%%%%%%%%%%%%%%%%%%%%%%%%%%%%%%%%%%%%%%%%%%%%%%%%%%%%%% 
\section{Infinitesimal  Poissonoid transformations and symmetries}
%%%%%%%%%%%%%%%%%%%%%%%%%%%%%%%%%%%%%%%%%%%%%%%%%%%%%%%%%%%%%%%%    
%
Let $ (M , \pi ,\mathcal{X}) $ be a locally Hamiltonian vector field with a locally defined Hamiltonian function $H$, and $ f _t $ a one-parameter group of  Poissonoid diffeomorphisms, then $ (f _t )_\ast \mathcal{X}= \pi ^\sharp \cdot \mathbf{d} K _t  $. This expression can be rewritten as 
\[
 \mathcal{X} = f _t ^\ast ( \pi) ^\sharp \cdot (f _t ^\ast \mathbf{d} K _t).  
\]
Let $ \xi $ be the vector field on $M$ that is the  infinitesimal generator of $ f _t $. Differentiating  the previous equation with respect to $t$ yields 

\begin{equation} %\label{eqn:infinitesimal_weakly_poissonoid} 
    \begin{aligned} 
   0& = \left .  \frac{ d } { dt }\right |_{t = 0 } \left[  f _t ^\ast ( \pi) ^\sharp \cdot (f _t ^\ast \mathbf{d} K _t)\right]\\
    & = (\mathcal{L} _{ \xi } \pi ^\sharp) \cdot (\mathbf{d} H) + \pi ^\sharp \cdot (\mathcal{L} _{ \xi } \mathbf{d} H) + \pi ^\sharp \cdot (\mathbf{d} \dot K)  \\
  & = \mathcal{L} _{ \xi } (\pi ^\sharp \cdot \mathbf{d} H) + \pi ^\sharp\cdot  (\mathbf{d} \dot K)   
\end{aligned}      
\end{equation} 
where we used the definition of Lie derivative and $ \dot K $ is the derivative of $ K $ with respect to $t$ computed at $ t = 0 $.
Hence
\[
     \mathcal{L} _{ \xi } (\pi ^\sharp \cdot \mathbf{d} H) =[\xi , \mathcal{X} ] = - \pi ^\sharp \cdot (\mathbf{d} \dot K) =  \pi ^\sharp \cdot (\mathbf{d} F)
\]
where $ F = - \dot K $.
A vector field $ \xi $ satisfying the equation above is called a {\bf infinitesimal  Poissonoid transformation}.

\begin{proposition}  \label{prop:poissonoid_basic} 
A vector field $ \xi $ is a infinitesimal  Poissonoid transformation for a locally Hamiltonian vector field $ \mathcal{X} $ if and only if $ [\xi , \mathcal{X} ]$  is a locally Hamiltonian  vector field. Moreover, if $ \xi $ is an infinitesimal symmetry  of $ \mathcal{X} $ then   it is also an infinitesimal   Poissonoid transformation. %The same result applies also in the case $ \xi $ is a Poisson infinitesimal symmetry of the system. 

\end{proposition}

\begin{proof}

      By definition of infinitesimal Poissonoid transformation we have  $[\xi , \mathcal{X} ]=  \pi ^\sharp \cdot (\mathbf{d} F)$, and thus $ \xi $ is locally Hamiltonian if and only if $ \xi $ is an infinitesimal  Poissonoid transformation. If  $ \xi $ is an infinitesimal symmetry  of $ \mathcal{X} $ then $ [\xi , \mathcal{X} ] = 0 $. Consequently, $ [\xi, \mathcal{X} ] $ is locally Hamiltonian with $ F $ the constant function, and thus it is an infinitesimal Poissonoid transformation.  
      \end{proof}
In \cite{carinena_canonoid_2013} canonoid transformations are studied using cohomology techniques. If $ (M , \pi, \mathcal{X})$ is a locally Hamiltonian system, then,  in analogy with \cite{carinena_canonoid_2013}  one can introduce  twisted boundary and  coboundary operators defined as follows:
\[
    \mathbf{\partial } _{ \mathcal{X} } = \mathbf{i} _{ \mathcal{X} } \circ \mathbf{d}\circ \mathbf{i} _{ \mathcal{X} },   
    \quad\quad \mathbf{d} _{ \mathcal{X} } = \mathbf{d}  \circ \mathbf{i} _{ \mathcal{X} } \circ \mathbf{d}  
\]
where $ \mathbf{d}  $ is the usual de Rham differential and $ \mathbf{i} _{ \mathcal{X} } $ denotes contraction. 
   \begin{proposition} Suppose $ (M , \pi , \mathcal{X} ) $ is a  locally Hamiltonian  system  with (locally defined) Hamiltonian $H$, then $ \mathcal{X} _F $ is a  Hamiltonian infinitesimal symmetry  of $ (M , \pi , \mathcal{X}) $ if and only if $ \mathbf{d} _{ \mathcal{X} } F \in \ker (\pi ^\sharp) $.  

       %Given a locally Hamiltonian dynamical system $ (M , \pi , \mathcal{X} ) $ with (locally defined) Hamiltonian $H$, let $ H _{ \mathcal{X} } ^0 (M) = \{ F \in C ^{ \infty } (M) | \mathbf{d} _{ \mathcal{X} }F= 0 \} $ be the zero cohomology group of $ \mathbf{d} _{ \mathcal{X} } $. Then $ H ^0 _{ \mathcal{X}}(M)  $ is a subset  of the set of Hamiltonian infinitesimal symmetries  of $ (M , \pi , \mathcal{X}) $. If $ \pi $ is nondegenerate then the two sets coincide.   
   \end{proposition} 
   \begin{proof}
       Let $F \in C ^{ \infty }(M)$, and let  $  \mathcal{X} _F = \pi ^\sharp \cdot \mathbf{d} F $  
       be its Hamiltonian vector field. Let $ \alpha $ be an arbitrary  one form, then the preceding equation can be written as
       \[
           \left\langle \alpha , \mathcal{X} _F \right\rangle = \left\langle \alpha , \pi ^\sharp \cdot  \mathbf{d} F \right\rangle. 
       \]
        Taking the Lie derivative of the left hand side yields  
        \begin{equation}\label{eqn:Lie_left} 
            \mathcal{L}_{ \mathcal{X} }  \left\langle \alpha , \mathcal{X}  _F \right\rangle = \left\langle \mathcal{L} _{ \mathcal{X} } \alpha , \mathcal{X}_F \right\rangle + \left\langle \alpha , \mathcal{L} _{ \mathcal{X}  }( \mathcal{X} _F)   \right\rangle. 
        \end{equation}         Taking the derivative of the right hand side yields
        \begin{equation}\label{eqn:Lie_right}\begin{aligned}             \mathcal{L} _{ \mathcal{X} }(\left\langle  \alpha  , \pi ^\sharp \cdot \mathbf{d} F \right\rangle  ) & = \mathcal{L} _{ \mathcal{X} } (\pi (\alpha , \mathbf{d} F))   \\
            & = (\mathcal{L} _{ \mathcal{X} } \pi) (\alpha , \mathbf{d} F)   + \pi (\mathcal{L} _{ \mathcal{X} }( \alpha) , \mathbf{d} F) + \pi (\alpha , \mathcal{L} _{ \mathcal{X} } (\mathbf{d} F))\\
      &   =  \pi (\mathcal{L} _{ \mathcal{X} }( \alpha ), \mathbf{d} F) + \pi (\alpha , \mathcal{L} _{ \mathcal{X} } (\mathbf{d} F))\\ 
      &   = \left\langle \mathcal{L} _{ \mathcal{X} } (\alpha) , \pi ^\sharp \cdot \mathbf{d} F \right\rangle    + \left\langle \alpha ,  \pi ^\sharp \cdot (\mathcal{L} _{ \mathcal{X} }  (\mathbf{d} F)) \right\rangle    
        \end{aligned} \end{equation} 
        since $ \mathcal{X} $ is locally Hamiltonian. Since $ \alpha $ is arbitrary, comparing \ref{eqn:Lie_left} and \ref{eqn:Lie_right} gives 
        \begin{equation}\label{eqn:Lie_end} 
            \mathcal{L} _{ \mathcal{X} } (\mathcal{X} _F) =[\mathcal{X} , \mathcal{X} _F ]=  \pi ^\sharp \cdot (\mathcal{L} _{ \mathcal{X} } (\mathbf{d} F)).      
        \end{equation} 
       Moreover, by Cartan's magic formula 
       \[ 
           \mathcal{L} _{ \mathcal{X} } (\mathbf{d} F) = \mathbf{d}  (\mathbf{i} _{ \mathcal{X} }( \mathbf{d} F)) + \mathbf{i} _{ \mathcal{X} } (\mathbf{d} ^2 F) =   \mathbf{d}  (\mathbf{i} _{ \mathcal{X} }( \mathbf{d} F)) = \mathbf{d} _{ \mathcal{X} } F
       \]   
      and hence 
      \begin{equation}\label{eqn:Lie_end1} 
          \mathcal{L} _{ \mathcal{X} } (\mathcal{X} _F) =[\mathcal{X} , \mathcal{X} _F ]=  \pi ^\sharp \cdot  (\mathbf{d}_{ \mathcal{X} }  F).      
      \end{equation} 
      Thus $ [ \mathcal{X} , \mathcal{X} _F ] = 0 $ if and only if $ \mathbf{d} _{ \mathcal{X} } F $ is in $ \ker ( \pi ^\sharp ) $. 
     Since $ \mathcal{X} $ is locally Hamiltonian we also have that $ \mathcal{L} _{ \mathcal{X} } H = 0 $ 
        %Thus, if  $ \mathcal{L} _{ \mathcal{X} } (\mathbf{d} F)= \mathbf{d}  _{ \mathcal{X} } F = 0 $, then $ \mathcal{L} _{ \mathcal{X} } (\mathcal{X} _F) = 0 $. Since $ \mathcal{X} $ is locally Hamiltonian we also have that $ \mathcal{L} _{ \mathcal{X} } H = 0 $.  %The converse is true when $ \pi $ is nondegenerate.  
    \end{proof}
    If $ \pi $ is non-degenerate and  $ H _{ \mathcal{X} } ^0 (M) = \{ F \in C ^{ \infty } (M) | \mathbf{d} _{ \mathcal{X} }F= 0 \} $ is the zero cohomology group of $ \mathbf{d} _{ \mathcal{X} } $, then $ H ^0 _{ \mathcal{X}}(M)  $ coincides with  the set of Hamiltonian infinitesimal symmetries  of $ (M , \pi , \mathcal{X}) $, reproducing the result given in \cite{carinena_canonoid_2013} for the symplectic case. 
    This result suggests that the cohomology approach introduced in \cite{carinena_canonoid_2013}, is not well adapted to the Poisson case. It may be possible to give a nice cohomological interpretation of symmetries in this case by using certain cohomology groups associated with a foliated space, called tangential cohomology  groups (see  \cite{moore_tangential_1988}) and references therein).%(see Global Analysis on Foliated Spaces Mathematical Sciences Research Institute Publications Volume 9, 1988, pp 68-91 Tangential Cohomology Calvin C. Moore, Claude Schochet) 
   
    Let $ \alpha $ be a differential one-form on $ (M , \pi)$. For the vector field associated to $ \alpha $ we use the following notation, $ \alpha ^\sharp : = \pi ^\sharp (\alpha) $.

   \begin{proposition}   Let $ (M , \pi  , \mathcal{X}) $ be a locally Hamiltonian system.  The vector field associated to the $ \mathbf{d} _{ \mathcal{X} } $ exact one-form $ \beta = \mathbf{d} _{ \mathcal{X} }  F $ is $ \beta ^\sharp = [\mathcal{X} , \mathcal{X} _F ] $, where $ \mathcal{X} _F $ is the Hamiltonian vector field of $F$.  
   \end{proposition} 
   \begin{proof}
     Equation   \ref{eqn:Lie_end} yields 
       \[
           \beta ^\sharp =[\mathcal{X} , \mathcal{X} _F ]=  \pi ^\sharp \cdot (\mathcal{L} _{ \mathcal{X} } (\mathbf{d} F)) = \pi ^\sharp \cdot \mathbf{d} _{ \mathcal{X} } F= \pi ^\sharp \cdot \beta .  
       \]  
   \end{proof} 

   \begin{proposition} Let $ (M , \pi , \mathcal{X}) $ be a locally-Hamiltonian system  and let $ \beta $ be a one-form  on $ (M , \pi) $. Then the vector field $ \beta ^\sharp = \pi ^\sharp \cdot \beta $ is an  infinitesimal Poissonoid transformation, if and only if $ \mathbf{d} _{ \mathcal{X} } \beta  = \mathbf{d} \alpha _1  $, where $ \alpha _1 \in \ker \pi ^\sharp $. In particular, if $ \mathbf{d} _{ \mathcal{X} }\beta  = 0 $, then $ \beta ^\sharp = \pi ^\sharp \cdot \beta $ is an infinitesimal Poissonoid transformation. If $ \pi $ is non-degenerate then   the vector field $ \beta ^\sharp = \pi ^\sharp \cdot \beta $ is an  infinitesimal Poissonoid transformation, if and only if $ \mathbf{d} _{ \mathcal{X} } \beta  = 0  $.
   \end{proposition} 
   \begin{proof} 
      $ \beta ^\sharp $ is infinitesimally Poissonoid if and only if  $ [ \beta ^\sharp , \mathcal{X} ] = \pi ^\sharp \cdot \mathbf{d} F $ for some (locally defined) function $F$. 
       The left hand side can be written as  
       \begin{align*}
           [ \beta ^\sharp, \mathcal{X} ] & = [ \pi ^\sharp \cdot \beta , \mathcal{X} ] \\
         & =  - \mathcal{L} _{ \mathcal{X} } (\pi ^\sharp \cdot  \beta) \\
        & = - \pi ^\sharp \cdot (\mathcal{L} _{ \mathcal{X} } \beta) \\
      & = - \pi ^\sharp \cdot  (\mathbf{d} (\mathbf{i} _{ \mathcal{X} } \beta) + \mathbf{i} _{ \mathcal{X} } (\mathbf{d} \beta ))  \\ 
       \end{align*}
      If $ (\mathbf{d} (\mathbf{i} _{ \mathcal{X} } \beta) + \mathbf{i} _{ \mathcal{X} } (\mathbf{d} \beta ))$ is closed, that is if, locally, we have  
      \[0=\mathbf{d} ((\mathbf{d} (\mathbf{i} _{ \mathcal{X} } \beta) + \mathbf{i} _{ \mathcal{X} } (\mathbf{d} \beta )))=
          \mathbf{d}^2  (\mathbf{i} _{ \mathcal{X} } \beta) + \mathbf{d} (\mathbf{i} _{ \mathcal{X} } (\mathbf{d} \beta ))=  \mathbf{d} (\mathbf{i} _{ \mathcal{X} } (\mathbf{d} \beta ))= \mathbf{d} _{ \mathcal{X} } \beta 
      \]
      then $ [ \beta ^\sharp , \mathcal{X} ] $ can be written, locally, as  $ \pi ^\sharp \cdot \mathbf{d} F $. The converse is not true unless $ \pi ^\sharp $ is non-degenerate. In order to do the general case let  $ \mathbf{i} _{ \mathcal{X} } (\mathbf{d} \beta) = \alpha _0 + \alpha _1 $ where $ \alpha _0 $ is closed and $ \alpha _1 \in \ker \pi ^\sharp $. This is equivalent to writing $ \mathbf{d} _{ \mathcal{X} } \beta = \mathbf{d} \alpha _1 $ with $ \alpha _1 \in \ker \pi ^\sharp $. Then there is a locally defined function $G$ such that we can write $ \alpha_0  = \mathbf{d} G $, and thus 
      \begin{align*} 
          [ \beta ^\sharp, \mathcal{X} ]&= - \pi ^\sharp \cdot  (\mathbf{d} (\mathbf{i} _{ \mathcal{X} } \beta) + \mathbf{i} _{ \mathcal{X} } (\mathbf{d} \beta ))\\
          & = - \pi ^\sharp \cdot  (\mathbf{d} (\mathbf{i} _{ \mathcal{X} } \beta+  G) + \alpha _1 )\\
          & = - \pi ^\sharp \cdot  (\mathbf{d} (\mathbf{i} _{ \mathcal{X} } \beta+  G)  )\\
          & = \pi ^\sharp \cdot \mathbf{d} F,
      \end{align*}
      with $F = - \mathbf{d} (\mathbf{i} _{ \mathcal{X} } \beta+  G)$.  
   \end{proof}
   This result generalizes the well known fact that infinitesimal canonical transformations are generated by closed forms in the de Rahm cohomology (that is the basis of the theory of generating functions) and the fact that infinitesimal canonoid transformations are generated by closed forms in the twisted cohomology introduced in \cite{carinena_canonoid_2013} to the case of Poissonoid transformations. In this case, however, not all the  infinitesimal Poissonoid transformations are generated by $ \mathbf{d} _{ \mathcal{X} } $-closed differential forms. 
   
   It seems possible to state the proposition above more elegantly by introducing  differential forms that are tangential to the foliation,  see \cite{moore_tangential_1988} and references therein for possible definitions of these forms.  
  
   %%%%%%%%%%%%%%%%%%%%%%%%%%%%%%%%%
   \subsection{Master Symmetries}
   %%%%%%%%%%%%%%%%%%%%%%%%%%%%%%%
   Let $ (M , \pi , \mathcal{X}) $ be a locally-Hamiltonian system. Recall that an  infinitesimal symmetry of $ \mathcal{X} $ is a vector field $ \xi $ such that $ [\xi , \mathcal{X} ]= 0$.  A  {\bf master symmetry} for $ \mathcal{X} $ is a vector field  $ \mathcal{X} $ such that 
   \[[\xi ,\mathcal{X} ]\neq 0, \quad\quad [[\xi , \mathcal{X} ], \mathcal{X} ] = 0.\]
    More generally, a {\bf master symmetry of degree $m$} for $ \mathcal{X} $ is a vector field  $ \mathcal{X} $ such that 
   \[[\xi ,\mathcal{X} ]\neq 0,\ldots,[\cdots [[\xi ,\underbrace{\mathcal{X} ],\mathcal{X} ],\cdots ,\mathcal{X} ] }_{m}\neq  0 ,~[\cdots [[\xi ,\underbrace{\mathcal{X} ],\mathcal{X} ],\cdots ,\mathcal{X} ] }_{m+1}= 0.
   \]
   The last condition in the equation above can also be written as $ \mathcal{L} _{ \mathcal{X} } ^{ m + 1 } (\xi)  = 0 $.
   %While Master symmetries do not, in general, generate Hamiltonian symmetries, they do when they are Poissonoid transformations:
   Neither master symmetries nor Poissonoid infinitesimal transformation do, in general, generate constants of motion. However, the next proposition shows that  some special infinitesimal Poissonoid transformations  generate constants of motion. 
   \begin{proposition}
        Suppose $ \xi $ is an infinitesimal Poissonoid transformation of the Hamiltonian system $ (M , \pi ,H) $, such that  the relationship  $ [ \xi , \mathcal{X}  _H ] = \pi ^\sharp \cdot \mathbf{d} F $ is  satisfied globally  and $ \mathcal{L} _{ [\xi , \mathcal{X} _H ]}H = 0 $. Then $ F $ is a constant of motion  of $ (M, \pi, H) $  and $ \xi $ is a master symmetry of degree 1.
       % $\xi$ is also a master symmetry of $H$ of degree 1, i.e. $ [[\xi,\mathcal{X}  _H ] , \mathcal{X}  _H ] = 0 $ then $F$ is a constant of motion.
    \end{proposition}   
    \begin{proof}
        Since  $ [ \xi , \mathcal{X}  _H ] = \pi ^\sharp \cdot \mathbf{d} F $ holds globally  and $ \mathcal{L} _{ [\xi , \mathcal{X} _H ]}H = 0 $ then $ \mathcal{X} _F = [ \xi , \mathcal{X}  _H ]  $ is a Hamiltonian symmetry of the system $ (M , \pi , H) $. Thus, by Noether's theorem, $ F $ is a constant of motion, that is $ \{ F, H \} = 0 $. It follows that,
        \[
           0= \{ F, H \} = [\mathcal{X} _F , \mathcal{X} _H ] = [[\xi , \mathcal{X} _H ], \mathcal{X} _H ],
        \] 
        and thus $ \xi $ is a master symmetry of degree 1. 
    \end{proof} 

   %%%%%%%%%%%%%%%%%%%%
   \iffalse 
   \begin{proposition}
    If $ \xi $ is an infinitesimal Poissonoid transformation of $H$ with $ [ \xi , \mathcal{X}  _H ] = \pi ^\sharp \cdot \mathbf{d} F $ and $\xi$ is also a master symmetry of $H$ of degree 1, i.e. $ [[\xi,\mathcal{X}  _H ] , \mathcal{X}  _H ] = 0 $ then $F$ is a constant of motion.    
   \end{proposition} 
   \begin{proof}[Attempting the proof]
       Let $ \mathcal{X} _F =  [ \xi , \mathcal{X}  _H ] = \pi ^\sharp \cdot \mathbf{d} F $ 
       Since $\xi$ is  a master symmetry of degree 1, and a master symmetry  we have
       \begin{align*} 0= [[ \xi , \mathcal{X}  _H ],\mathcal{X} _H ] & = [ \pi ^\sharp \cdot \mathbf{d} F  , \mathcal{X}  _H ] = [\mathcal{X} _F, \mathcal{X} _H ] = \mathcal{L} _{ \mathcal{X} _F } \mathcal{X} _H \\ 
           &=   \mathcal{L} _{ \mathcal{X} _F } (\pi ^\sharp \cdot \mathbf{d} H) \\
           & = (\mathcal{L} _{ \mathcal{X} _F } \pi ^\sharp)\cdot \mathbf{d} H + \pi ^\sharp \cdot (\mathbf{d} (\mathcal{L} _{ \mathcal{X} _F } H))    \\
     & =  \pi ^\sharp \cdot (\mathbf{d} (\mathcal{L} _{ \mathcal{X} _F } H))    \end{align*} 
 This means $ \mathbf{d} (\mathcal{L} _{ \mathcal{X} _F } H) \in \ker \pi ^\sharp $. Even in the non-degenerate case this does not imply that 
$ \mathcal{L} _{ \mathcal{X} _F } H=0$, but only that $ \mathbf{d} (\mathcal{L} _{ \mathcal{X} _F } H)= 0 $.  
  \end{proof} 
  \fi
 %%%%%%%%%%%%%%%%%%%%%%%%%%%%%%%%%%%%%%%%%%%%%%%%%%%%%%%%%%%%%%%%%%%%%%%%%%%%%%%% 
    Let $ (M , \pi , \mathcal{X}) $ be a locally-Hamiltonian system. A function $ T $    is called a {\bf generator of constants of motion of degree $m$ } for $ \mathcal{X} $  if 
   \[
       \mathcal{L} _{ \mathcal{X} } T \neq 0 , \ldots ,~ \mathcal{L} _{ \mathcal{X} } ^m T \neq 0 ,~ \mathcal{L} _{ \mathcal{X} } ^{ m + 1 } T = 0. 
   \]

   Suppose $ \xi $ is a Hamiltonian vector field with Hamiltonian function $ T $, with $ \xi = \pi ^\sharp \cdot \mathbf{d} T $. 
   Taking the Lie derivative of both sides yields 
   \[\mathcal{L} _{\mathcal{X} }(\xi) =   \mathcal{L} _{ \mathcal{X} } (\pi ^\sharp \cdot \mathbf{d} T)  =  \mathcal{L} _{ \mathcal{X} } (\pi ^\sharp) \cdot \mathbf{d} T +   \pi ^\sharp \cdot(\mathcal{L} _{ \mathcal{X} } ( \mathbf{d} T)) =   \pi ^\sharp \cdot \mathbf{d} (\mathcal{L} _{ \mathcal{X} } ( T)).  \]
   This relation can be generalized to 
   \[
       \mathcal{L} _{ \mathcal{X} } ^{ m + 1 } (\xi) = \pi ^\sharp \cdot \mathbf{d} (\mathcal{L} _{ \mathcal{X} } ^{ m + 1 }( T)).    
   \]
   We call $T$ {\bf the generator of a Hamiltonian master symmetry of degree $ m  $}
   if and only if $ \mathbf{d} (\mathcal{L} _{ \mathcal{X} } ^{ m + 1 } (T))  \in \ker \pi ^\sharp $.
   Note that, in particular, If $ T $ is the generator of constants of motion of degree $m$, then $ \mathcal{L} _{ \mathcal{X} } ^{ m + 1 } T = 0$, and thus $ \mathbf{d} ( \mathcal{L} _{ \mathcal{X} } ^{ m + 1 } T) = 0 $. It follows that  $ \mathcal{L} _{ \mathcal{X} } ^{ m + 1 } (\xi) = 0 $, and thus $ \xi $ is a master symmetry of degree $m$. The converse is clearly not true if $ \pi ^\sharp $ is degenerate.  
   The following proposition shows that Poissonoid transformations are very general: every master symmetry, and also every infinitesimal symmetry of the vector field can be generated using Poissonoid transformations.

   %More generally we have that  $ \xi $ is a master symmetry of degree $m$, that is  $ \mathcal{L} _{ \mathcal{X} } ^{ m + 1 } (\xi) = 0 $, if and only if $ \mathbf{d} (\mathcal{L} _{ \mathcal{X} } ^{ m + 1 } (T)) $ is in the kernel of $ \pi ^\sharp  $.   
   \begin{proposition}
       Let $ (M , \pi , \mathcal{X}) $ be a locally-Hamiltonian vector field. Suppose that $ \xi $ is a master symmetry of degree $ m + 1 $. Then  the vector field $ \mathcal{L} _{ \mathcal{X} } ^m (\xi) $  is an infinitesimal Poissonoid transformation.
   \end{proposition}
   \begin{proof}
      Since $ \xi $ is a master symmetry of degree $ m + 1 $, we have
      \[
           \mathcal{L} _{ \mathcal{X} } ^{ m + 1 } (\xi)  =  [\mathcal{L} _{ \mathcal{X} } ^{ m   } (\xi), \mathcal{X} ]  = 0. 
      \] 
     Then  $ \mathcal{L} _{ \mathcal{X} } ^m (\xi) $  is an infinitesimal Poissonoid transformation by Proposition \ref{prop:poissonoid_basic}.
   \end{proof} 

 \begin{proposition} Let $ (M , \pi , \mathcal{X}) $ be a locally-Hamiltonian system  and let $ \beta $ be a one-form  on $ (M , \pi) $. Suppose the vector field $ \beta ^\sharp = \pi ^\sharp \cdot \beta $ is an infinitesimal Poissonoid transformation such that $ \mathbf{d} _{ \mathcal{X} } \beta = 0 $.    Then $\beta ^\sharp $ is master symmetry of degree $m$ ( $ m \geq 1 $) if and only if $ \mathbf{i} _{ \mathcal{X} } \beta $ is the generator of a Hamiltonian master symmetry of degree $ m - 1 $.
 \end{proposition}
   \begin{proof}
        $ \beta ^\sharp $ is a master symmetry of degree $m$ if and only if 
        \begin{equation}\label{eqn:master_symmetry} 
           \mathcal{L} _{ \mathcal{X} } ^{ m + 1 } \beta  ^\sharp = 0, \quad \mbox{and }\quad  \mathcal{L} _{ \mathcal{X} } ^m \beta ^\sharp \neq 0.
       \end{equation} 
       Since $ \beta ^\sharp = \pi ^\sharp \cdot \beta $, we have 
       \begin{align*} 
           \mathcal{L} _{ \mathcal{X} } ^{ m + 1 } (\beta ^\sharp)&  = \mathcal{L} _{ \mathcal{X} } ^{ m + 1 } (\pi ^\sharp \cdot \beta)   =   \mathcal{L} _{ \mathcal{X} } ^{ m } (\mathcal{L} _{ \mathcal{X} } (\pi ^\sharp \cdot \beta))\\ &  =   \mathcal{L} _{ \mathcal{X} } ^{ m } (\mathcal{L} _{ \mathcal{X} } (\pi ^\sharp) \cdot \beta + \pi ^\sharp \cdot (\mathcal{L} _{ \mathcal{X} } (\beta))  ) \\ & =   \mathcal{L} _{ \mathcal{X} } ^{ m } ( \pi ^\sharp \cdot (\mathcal{L} _{ \mathcal{X} } (\beta))  )\\
           & = \cdots \\
           & = \pi ^\sharp \cdot(  \mathcal{L} _{ \mathcal{X} } ^{ m + 1 } \beta ),
       \end{align*}
      where we used that $ \mathcal{L} _{ \mathcal{X} } \pi ^\sharp = 0 $ since $ \mathcal{X} $ is locally Hamiltonian.  Hence, equation \ref{eqn:master_symmetry} can be equivalently written as
      \[
           \mathcal{L} _{ \mathcal{X} } ^{ m + 1 } \beta  ^\sharp =  \pi ^\sharp \cdot (\mathcal{L} _{ \mathcal{X} } ^{ m + 1 } \beta)  = 0, \quad \mbox{and }\quad  \mathcal{L} _{ \mathcal{X} } ^m \beta ^\sharp =  \pi ^\sharp \cdot (\mathcal{L} _{ \mathcal{X} } ^{ m  } \beta) \neq 0,
      \]
      or 
      \[
            (\mathcal{L} _{ \mathcal{X} } ^{ m + 1 } \beta) \in \ker \pi ^\sharp  , \quad \mbox{and }\quad  (\mathcal{L} _{ \mathcal{X} } ^{ m  } \beta) \notin \ker \pi ^\sharp. 
      \]
      Using Cartan's magic formula yields 

      \begin{align*} 
          \mathcal{L} _{ \mathcal{X} } ^{ m + 1 } \beta & = \mathcal{L} _{ \mathcal{X} } ^m (\mathcal{L} _{ \mathcal{X} } \beta) = \mathcal{L} _{ \mathcal{X} } ^m (\mathbf{d} (\mathbf{i} _{ \mathcal{X} } \beta) + \mathbf{i} _{ \mathcal{X} } \mathbf{d} \beta)\\
      &     = \mathbf{d} (\mathcal{L} _{ \mathcal{X} } ^m (\mathbf{i} _{ \mathcal{X} }\beta )  ) + \mathcal{L} _{ \mathcal{X} } ^m (\mathbf{i} _{ \mathcal{X} } \mathbf{d} \beta)  \\
    &  =  \mathbf{d} (\mathcal{L} _{ \mathcal{X} } ^m (\mathbf{i} _{ \mathcal{X} }\beta )  )+ \mathcal{L} _{ \mathcal{X} } ^{m-1} (\mathbf{i} _{ \mathcal{X} } \mathbf{d} ( \mathbf{i} _{ \mathcal{X} }\mathbf{d} \beta  )+ \mathbf{d} (\mathbf{i} _{ \mathcal{X} } \mathbf{i} _{ \mathcal{X} }( \mathbf{d} \beta)) ) \\
   & =   \mathbf{d} (\mathcal{L} _{ \mathcal{X} } ^m (\mathbf{i} _{ \mathcal{X} }\beta )  )+ \mathcal{L} _{ \mathcal{X} } ^{m-1} (\mathbf{i} _{ \mathcal{X} } (\mathbf{d} _{ \mathcal{X} } \beta)   )\\
  & =  \mathbf{d} (\mathcal{L} _{ \mathcal{X} } ^m (\mathbf{i} _{ \mathcal{X} }\beta )  )
      \end{align*} 
      where we used the fact that  $ \mathbf{i} _{ \mathcal{X} } \mathbf{i} _{ \mathcal{X} } \alpha = 0 $ for any  $k$-form $ \alpha $, and the fact that $ \mathcal{X} $ is an  infinitesimal Poissonoid transformations such that $ \mathbf{d} _{ \mathcal{X} } \beta = 0 $.     It follows that $ \mathcal{L} _{ \mathcal{X} } ^{ m + 1 } \beta ^\sharp = 0 $ if and only if $ \mathbf{d} (\mathcal{L} _{ \mathcal{X} } ^{ m  } (\mathbf{i} _{ \mathcal{X} } \beta ))  \in \ker \pi ^\sharp $, that is if and only if $ \mathbf{i} _{ \mathcal{X} } \beta $ is the generator of a Hamiltonian master symmetry of degree $ m -1$. 

   \end{proof}

\section*{Acknowledgments}
The authors wish to thank Luca Degiovanni, Francesco Fass\`o, and Anton Izosimov for comments and suggestions, and Shengda Hu for  several discussions concerning this work. We would also  like to thank both  reviewers for their insightful comments on the paper, as these comments led us to an improvement of our work.  The second author was supported  by an NSERC Discovery Grant.

%\printbibliography
% Fakesection
%-------------------------------------------
\bibliographystyle{aims}
%\bibliographystyle{acm}
%\bibliographystyle{amsplain}
%\bibliographystyle{unsrt}   % this means that the order of references is determined by the order in which the citations appear in the text.
%\nocite{*} 		% The command \nocite{*} causes all items in the database to be included in the references, regardless of whether or not they are cited in the paper.
\bibliography{Books,Papers}% list here all the bibliographies that  you need.
%\AtEveryBibitem{\clearlist{language}}
%------------------------------------------

\end{document}

\subsection{Infinitesimal (weakly) Poissonoid transformations}
A diffeomorphism $ f:M \to M $ is weakly Poissonoid with respect to $ \mathcal{X} $, when $\mathcal{X} $ is a Poisson vector field with respect to $ f ^\ast \pi $, that is if and only if $ \mathcal{L} _{ \mathcal{X} } (f ^\ast \pi) = 0 $. If we consider a one-parameter group of weakly Poissonoid diffeomorphisms, then   $ \mathcal{L} _{ \mathcal{X} } f_t  ^\ast (\pi) = 0 $. Let $ \xi $ be the vector field on $M$ that is the  infinitesimal generator of $ f _t $. Differentiating with respect to $t$ yields 
\begin{equation} \label{eqn:infinitesimal_weakly_poissonoid}
    \left .  \frac{ d } { dt }\right |_{t = 0 }  \mathcal{L} _{ \mathcal{X} } (f_t  ^\ast \pi) = \mathcal{L} _{ \mathcal{X} } \left(\left.\frac{ d } { dt }\right | _{ t = 0 } ( f _t ^\ast \pi) \right)  = \mathcal{L} _{ \mathcal{X} } \mathcal{L} _{ \xi } \pi  = 0    
\end{equation} 
where the last equality follows from the definition of Lie derivative. 
A vector field satisfying \eqref{eqn:infinitesimal_weakly_poissonoid} is called a {\bf infinitesimal weakly Poissonoid transformation}
\begin{proposition}
A vector field $ \xi $ is the infinitesimal generator of a weakly Poissonoid transformation for a locally Hamiltonian vector field $ \mathcal{X} $ if and only if $ [\xi , \mathcal{X} ]$  is a Poisson vector field.
\end{proposition}

\begin{proof}

     Since $ \xi $ is  the infinitesimal generator of a weakly Poissonoid transformation for 
    $ \mathcal{X} $ we have that $ \mathcal{L} _{ \mathcal{X} } \mathcal{L} _{ \xi } \pi=0 $. Moreover,  $ \mathcal{L} _{ \mathcal{X}  }\pi = 0 $, since $ \mathcal{X} $ is locally Hamiltonian. Hence, we have 
    \[
        \mathcal{L} _{ [\xi , \mathcal{X} ]} \pi = \mathcal{L} _{ \xi } \mathcal{L} _{ \mathcal{X}  }\pi - \mathcal{L} _{ \mathcal{X} } \mathcal{L} _{ \xi } \pi =0
    \]
    and thus $ [ \xi , \mathcal{X} ] $ is a Poisson vector field. 

\end{proof}